\documentclass[preprint]{elsarticle}
\pdfoutput=1

\usepackage{amsfonts,amsmath,bm}
\usepackage[vlined,boxruled,titlenumbered]{algorithm2e}
\usepackage{color}
\usepackage{hyperref}

\newtheorem{defn}{Definition}
\newtheorem{lemma}{Lemma}

\newenvironment{proof}{{\em Proof.}}{$\Box$}

\newcommand{\cnn}{\mathbb{C}^{n\times n}}
\newcommand{\cn}{\mathbb{C}^{n}}
\newcommand{\mathC}{\mathbb{C}}
\newcommand{\x}{\times}

\newcommand{\eqnref}[1]{(\ref{#1})}

\newcommand{\Dw}{\ensuremath{D_\text{w}}}

\DeclareMathOperator{\range}{range}

\DeclareMathOperator{\argmin}{argmin}

\begin{document}

\title{ Deflation and Flexible SAP-Preconditioning of GMRES in Lattice QCD Simulation}

\author[wup]{Andreas Frommer}
\author[reg]{Andrea Nobile}
\author[wup]{Paul Zingler}

\address[reg]{Institute for Theoretical Physics, University of
  Regensburg, 93040 Regensburg, Germany}
\address[wup]{Department of Mathematics, University of Wuppertal,
  42097 Wuppertal, Germany}

\tnotetext[t1]{Supported by DFG collaborative research center
 SFB/TR-55 ``Hadron Physics from Lattice QCD''.}

\begin{abstract}
  The simulation of lattice QCD on massively parallel computers stimulated the 
  development of scalable algorithms for the solution of sparse linear systems.
  We tackle the problem of the Wilson-Dirac operator inversion by 
  combining a Schwarz alternating procedure (SAP) in multiplicative form 
  with a flexible variant of the GMRES-DR algorithm.   
  We show that restarted GMRES is not able to converge when the system 
  is poorly conditioned. By adding deflation in the form of the FGMRES-DR 
  algorithm, an important fraction of the information produced by the iterates
  is kept between successive restarts leading to convergence in cases in which 
  FGMRES stagnates.
  
\end{abstract}

\maketitle

\section{Introduction}

The simulation of Lattice QCD is among one of the most challenging problems in computational science because of its enormous computational cost, see \cite{gatt09}, e.g.. 
With sufficiently powerful computers becoming available during the last fifteen years, the simulation of Lattice QCD including dynamical fermions became a reality.
The de-facto standard algorithm used to accomplish this task is the Hybrid Monte Carlo (HMC) algorithm \cite{KennedyDuane}.
HMC introduces a fictious time, in which a dynamical system is evolved according to the system Hamiltonian.
In order to integrate the Hamiltonian equations of motion, the forces due to both, gauge field and pseudo-fermion field,
have to be evaluated accurately.
The by far most demanding numerical task is the inversion of the fermionic Dirac matrix, needed to evaluate the fermion force. In this paper we investigate a class of inverters for the fermionic Dirac matrix which is particularly well adapted 
to current high performance hardware architecture. Our work was primarily inspired by the architecture of the QPACE machine, see \cite{qpace09,qpace10}, but it is of a general nature and applies to all architectures were data movements rather than arithmetic operations tend to be the limiting factor on performance.   

The development of machine-aware algorithms has become particularly important during the last years as computer architectures move towards dramatically increased floating point (FP) speed and increased on-chip parallelism.
For reasons of cost and because of physical limitations, these improvements are not accompanied by comparable 
progresses in the bandwidths and latencies of main memory and processor interconnects. It is predicted that this
situation will be even aggravated as we move towards the Exascale area, see \cite{OppChaExa}, e.g.. 

An often used rough indicator of balance between processor and memory subsystems is the ratio between peak main memory bandwidth and peak floating point performance. In the past few years this value went down from  $\approx 1$ to $\approx 1/8$ and less for the CELL processor, modern GPUs and general purpose x86 systems. This increasing gap is the main reason for a line of computer science research on software techniques aimed at maximizing cache utilization. A famous example is \cite{FFTW}.

The algorithm cost evaluation must take into account not only the FP operations (as it was traditionally done some years ago when FP was the most relevant bottleneck in numerical applications), but has to carefully evaluate the amount and pattern of data movement among the levels of memory hierarchy and the network as these are now the most limiting factors.

Algorithms that are more suitable for these increasingly unbalanced computer architectures are particularly attractive as they can give easier access to the FP performances offered today.        

We describe an algorithm for the Dirac-Wilson matrix inversion that is able to exploit the architectural features of modern machines. 
We show the convergence behaviour and demonstrate the performance with our implementation on the peculiar SFB-TR55 QPACE architecture based on the IBM PowerXCell8i processor with state of the art lattices.

 \section{Domain decomposition and SAP} 
The Wilson fermion matrix $\Dw$ \cite{Wilson:1975id} describes a (periodic) nearest neighbor
coupling on the 4-dimensional
space-time lattice.  Without further explanation, our desire is to solve the linear system 
\[
\Dw x = b.
\]
Owed to the complexity and size of $\Dw$, a direct solution is not feasible and we therefore exploit the structure of $\Dw$.
We decompose the overall lattice into sublattices, called {\em domains}. For each
domain, we obtain a local Wilson fermion operator by simple restriction; at the local boundary 
of a domain, we remove the couplings with neighboring domains. To describe this formally, let $n = (n_1,n_2,n_3,n_4)$
stand for a 4-tupel describing a lattice site and let $L = \{ n: (0,0,0,0) \leq n \leq (d_1-1,d_2-1,d_3-1,d_4-1)\}$ denote
the overall lattice. The Dirac Wilson operator can then be written as
\begin{equation}
\label{Dw}
[D_\text{w}]_{nm} = 
\delta_{n,m}
   - \kappa \left( \sum_{j=1}^4  (1+\gamma_j) U_{n,j} \delta_{n+\hat j,m} 
      +  (1-\gamma_j)
        U^\dagger_{n-\hat j,j} \delta_{n-\hat j,m} \right), \, n,m \in L,
\end{equation}
where $n \pm \hat j$ is to be taken mod $d_j$.

A domain $L^\ell$ is characterized by the bounds $a_j^\ell, d_j^\ell$ for its lattice sites which are
given by $L^\ell = \{ n: (a^\ell_1,a^\ell_2,a^\ell_3,a^\ell_4) \leq n \leq  (d^\ell_1-1,d^\ell_2-1,d^\ell_3-1,d^\ell_4-1) \}$.
The local Dirac operator for domain $L^\ell$ is thus given as
\begin{equation}
\label{Dw_local}
[D^\ell_\text{w}]_{nm} = 
\delta_{n,m}
   - \kappa \left( \sum_{j=1}^4 {}'  (1+\gamma_j) U_{n,j} \delta_{n+\hat j,m} 
      +  (1-\gamma_j)
        U^\dagger_{n-\hat j,j} \delta_{n-\hat j,m} \right), n,m \in L^\ell,
\end{equation}
where $\sum '$ indicates that the sum is to be taken only over those contributions for 
which $n+\hat j$ or $n - \hat j$ is in $L^\ell$.

Assuming that the local Wilson operators are easy to invert, we can perform the
following basic domain decomposition 
iteration, known as the multiplicative Schwarz method in the domain decomposition literature \cite{hackbusch94,smith96}.  
\begin{algorithm}
  \caption{Multiplicative Schwarz method \label{ms:alg}}
  \DontPrintSemicolon
choose intitial guess $x$ \;
\Repeat{until convergence} {
 \For{$\ell=1,\ldots,p$}{
   compute residual $r^\ell$ for domain $\ell$: $r^\ell = (b-\Dw x)^\ell$ \;
   solve $\Dw^\ell y^\ell = r^\ell$\;
   update: replace part $x^\ell$ in $x$ by $y^\ell$ \; 
   }
}
\end{algorithm}

Herein, the superscript $\ell$ indicates that we only take the part which corresponds to the subdomain $\ell$.
Due to the nearest neighbor coupling of $\Dw$, computing $r^\ell$ is a local operation which 
involves only those components of $x$ which belong to subdomain $L^\ell$, i.e. $x^\ell$ and those from
neighboring subdomains which are at the boundary to domain $L^\ell$.  

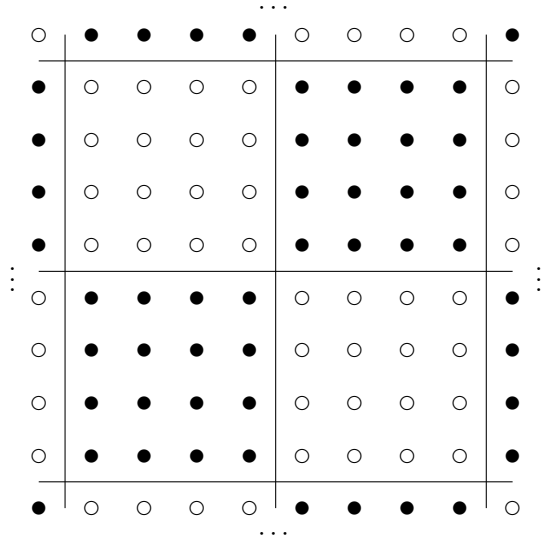
\begin{figure}
\begin{center}
\setlength{\unitlength}{0.7cm}
\begin{picture}(10,10)(-1,-1)
\multiput(0,0)(1,0){4}{\multiput(0,0)(0,1){4}{\circle*{0.25}}}
\multiput(4,4)(1,0){4}{\multiput(0,0)(0,1){4}{\circle*{0.25}}}
\multiput(0,8)(1,0){4}{\multiput(0,0)(0,1){1}{\circle*{0.25}}}
\multiput(4,-1)(1,0){4}{\multiput(0,0)(0,1){1}{\circle*{0.25}}}
\multiput(-1,4)(1,0){1}{\multiput(0,0)(0,1){4}{\circle*{0.25}}}
\multiput(8,0)(1,0){1}{\multiput(0,0)(0,1){4}{\circle*{0.25}}}
\multiput(0,4)(1,0){4}{\multiput(0,0)(0,1){4}{\circle{0.25}}}
\multiput(4,0)(1,0){4}{\multiput(0,0)(0,1){4}{\circle{0.25}}}
\multiput(4,8)(1,0){4}{\multiput(0,0)(0,1){1}{\circle{0.25}}}
\multiput(0,-1)(1,0){4}{\multiput(0,0)(0,1){1}{\circle{0.25}}}
\multiput(-1,0)(1,0){1}{\multiput(0,0)(0,1){4}{\circle{0.25}}}
\multiput(8,4)(1,0){1}{\multiput(0,0)(0,1){4}{\circle{0.25}}}
\put(-1,-0.5){\line(1,0){9}}
\put(-1,3.5){\line(1,0){9}}
\put(-1,7.5){\line(1,0){9}}
\put(3.5,-1){\line(0,1){9}}
\put(-0.5,-1){\line(0,1){9}}
\put(7.5,-1){\line(0,1){9}}

\put(-1,8){\circle{0.25}}
\put(-1,-1){\circle*{0.25}}
\put(8,-1){\circle{0.25}}
\put(8,8){\circle*{0.25}}

\put(-1.5,3.5){\makebox(0,0){$\vdots$}}
\put(8.5,3.5){\makebox(0,0){$\vdots$}}
\put(3.5,-1.5){\makebox(0,0){$\cdots$}}
\put(3.5,8.5){\makebox(0,0){$\cdots$}}
\end{picture}
\end{center}
\caption{domain decomposition in 2 dimensions \label{red-black:fig}}
\end{figure}

If we use a red-black coloring of the subdomains as depicted in Fig.~\ref{red-black:fig}
for a two-dimensional lattice, and we have an even number of sub-domains in each dimension of the lattice, then there 
is no coupling between subdomains of the same color. When we order all red (say) domains first, we see that in the
algorithm above the values of the residual $r^\ell$ does not depend on whether or not we updated the part $x^j$ of the iterate
for the other subdomains of the same color. If we proceed by colors, we end up with the following variant of
the basic domain decomposition iteration which is termed the Schwarz alternating procedure (SAP) in \cite{Luescher2003}.

\begin{algorithm}
  \caption{SAP iteration \label{sap:alg}}
  \DontPrintSemicolon
  \DontPrintSemicolon
choose intitial guess $x$ \;
\Repeat{until convergence} {
 \For{all red domains $\ell$}{
   compute residual $r^\ell$ for domain $\ell$: $r^\ell = (b-\Dw x)^\ell$ \;
   solve $\Dw^\ell y^\ell = r^\ell$\;}
 update: replace part $x^\ell$ in $x$ by $y^\ell$ for all red domains $\ell$ \;
\For{all black domains $\ell$}{
   compute residual $r^\ell$ for domain $\ell$: $r^\ell = (b-\Dw x)^\ell$ \;
   solve $\Dw^\ell y^\ell = r^\ell$\;}
 update: replace part $x^\ell$ in $x$ by $y^\ell$ for all black domains $\ell$ \;
}
\end{algorithm}

\section{SAP in computational practice}
In computational practice, the SAP method can rarely be used directly in the form given by Algorithm~\ref{sap:alg}.
On one hand it is usually too costly to solve the local systems $\Dw^\ell y^\ell = r^\ell$ exactly.
One rather computes an approximation to $y^\ell$ using another, {\em inner} iterative method on the systems
$\Dw^\ell y^\ell = r^\ell$. Usually, the fastest overall method arises if one requires a fairly low accuracy in the inner iteration. For example, one requires that the approximate solution $\tilde{y}^\ell$
reduces the initial residual by just a factor of 10, i.e.\ one requires
\[
\| r^\ell - \Dw^\ell \tilde{y}^\ell\| \leq 0.1 \cdot \| r^\ell \|.
\]
Such a method is termed an {\em inexact} SAP method.
 
On the other hand, even the exact SAP iteration may diverge for the Wilson Dirac operator when the hopping parameter $\kappa$ approaches
the critical value $\kappa_c$. Indeed, most of the convergence theory for multiplicative Schwarz methods assumes that the operator
is hermitian and positive definite, see \cite{hackbusch94,smith96}, which is not the case for the Wilson Dirac operator. 

The usage of SAP methods is particularly suitable on parallel machines and in general on machines with a memory hierarchy.
Provided that at least two subdomains are assigned to a machine node, there is no need for network communication during the solution of the Wilson equation on the subdomains, since all the necessary data is on the node.
Data only needs to be communicated for the update of the residual.
In practice, there are more than two subdomains per node giving the possibility to overlap network communication and computation.
The required network bandwidth, and in particular latency, is less stringent than for a simple Krylov solver, since the approximate solution of the subdomain Wilson equation requires some iterations.   
The required machine performance on global network operations (i.e. global sum, barrier, etc.) is also typically reduced, since the frequency of such operations is lower than in the case of a Krylov solver.
We have chosen SAP as a preconditioner to be used on the most time-critical solves on our machine QPACE.
The QPACE torus network \cite{qpace10} is designed with a typical Krylov solver in mind and has therefore high bandwidth and low latency, nevertheless SAP is particularly appealing due to the peculiar memory hierarchy architecture of the IBM PowerXCell 8i Processor.
On this processor, the fast on-chip memory is not managed automatically by the hardware as on standard processors (cache), but its control is completely left to the programmer.
If the size of the subdomain is chosen such that it fits the fast memory, no main memory access is necessary for the subdomain solve, giving the possibility to achieve impressive sustained performance for the subdomain solves (up to 50\% of the peak) \cite{nobile}, \cite{naka_et_al11}.  
This feature of the SAP algorithm is interesting also for future architectures as the main memory performance will continue to improve at a lower rate than processor performance.

\section{Domain decomposition preconditioner} \label{short:sec}

Exact or inexact SAP iterations can be used as a preconditioner in a Krylov subspace method like GMRES 
\cite{Saad-Schultz:GMRES} or GCR \cite{Axelsson:1987}.  
This will usually result in an iterative process which converges faster than the Krylov subspace method without preconditioning or the stand-alone SAP iteration. We will discuss GMRES and GCR in quite some detail
later in section \ref{sec:restart}. At this point, we concentrate on aspects pertaining to 
the use of {\em inexact} SAP iterations in the context of preconditioning.

Loosely speaking, the speed of convergence of a Krylov subspace method for a generic system
\[
Ax = b
\]
gets faster the closer $A$ is to the identity. So, to accelerate the convergence one tries to 
construct (or implicitly use) a matrix $M$, the preconditioner, such that the matrix $AM$ of the (right) preconditioned
system
\[
(AM)y = b, \enspace x = My
\]
is close to the identity. In this sense, $M$ should be an approximation to the inverse of $A$.
Such $M$ is given implicitly if we perform some steps of an iterative procedure like SAP.

Each step of the Krylov subspace method requires one multiplication with $A$ and one with $M$.
If we perform exact SAP with a fixed number, $s$ say, of iterative steps, the preconditioner $M$
can be expressed as the truncated series
\[
M = \sum_{j=0}^{s-1} (A_L^{-1}A_U)A_L^{-1},
\]
where the matrices $A_L$ and $A_U$ result from a splitting $A = A_L - A_U$ of the matrix $A$. 
To be specific, for the SAP iteration from Algorithm~\ref{sap:alg}, where $A = \Dw$, this splitting is
based on the red-black permuted matrix $\Dw$ given as
\[
\Dw = \left[ \begin{array}{cc} \Dw^{bb} & \Dw^{br} \\
                                \Dw^{rb} & \Dw^{rr}
        \end{array}
      \right] .      
\]
Herein, the blocks $\Dw^{bb}$ and $\Dw^{rr}$ are made up of all the local Dirac operators $\Dw^\ell$
of the black and the red subdomains, respectively and the off-diagonal blocks contain the couplings between subdomains of different color. Then
\[
A_L = \left[ \begin{array}{cc} \Dw^{bb} & 0 \\
                                \Dw^{rb} & \Dw^{rr}
        \end{array}
      \right],  \enspace
A_U = \left[ \begin{array}{cc}  0 & - \Dw^{br} \\
                                0 &  0
        \end{array}
      \right].
\]

In {\em inexact} SAP we invert $A_L$ only approximately, using some inner iteration on the diagonal blocks $\Dw^\ell$. If this inner iteration is non-stationary, as it is the case if we perform some steps of GMRES or MR \cite{saad_book}, the implicitly given preconditioner changes at each iteration. This is a fact one has to account for when formulating the preconditioned Krylov subspace methods. In the linear algebra literature such methods are termed {\em flexible}, see \cite{Saad:FGMRES,saad_book} since they adapt themselves to situations where the preconditioner varies from one iteration to the next.

The next chapter will first introduce the flexible restarted GMRES method before turning to the crucial subject of this paper, namely its acceleration via deflation. 

\section{F-GMRES-DR}
\label{sec:restart}

To simplify the presentation we will use a generic notation in this chapter, i.e.\ the linear system to
solve will be denoted
\[
Ax = b.
\]
A preconditioner will be denoted $M$. In a variable preconditioning context, where the preconditioner will depend on the current iterative step $j$, we denote the preconditioner 
by $M_j$. Note that $M_j$ will usually not be given explicitly but rather be the result of
some steps of a non-stationary iteration like inexact SAP. 

Particular aspects for (inexact) SAP preconditioning of the Wilson Dirac operator will be discussed
later. 

\subsection{GMRES and GCR}
We consider a generic non-singular linear system
\[
A x = b \mbox{ with solution } x_* = A^{-1}b.
\]
Given an initial guess $x_0$ and its residual $r_0 = b-Ax_0$, the $j$-th Krylov 
subspace $K_j(A,r_0)$ is defined as
\[
K_j(A,r_0) = \mbox{span}\{r_0,Ar_0,\ldots,A^{j-1}r_0\}.
\]
For ease of notation, we will often write $K_j$ instead of $K_j(A,r_0)$.
It is known that the solution $x_* = A^{-1}b$ is contained in the space $x_0 + K_{j_*}$
where $j_*$ denotes the smallest $j$ for which $K_j = K_{j+1}$. Krylov subspace 
methods work by iteratively choosing approximations $x_j$ from $x_0 + K_j$ to approximate
$x_*$. The GMRES and GCR method are Krylov subspace methods which both obtain 
the ``optimal'' $x_j$ in the sense that the residual $r_j = b-Ax_j$ satisfies the minimal residual condition
\[
   \|r_j\|  = \min_{x \in x_0+K_j} \| b-Ax \|.
\]
Note that for any subspace $\mathcal{V}$ of $\mathC^n$ we have 
\begin{equation} \label {gen_minimal_residual:eq}
  x=x_0 + \delta x  \mbox{ minimizes } \|b-Ax\| \mbox{ over } x_0 + \mathcal{V} \, 
 \Leftrightarrow \,  r_0 -Ax \perp A \mathcal{V},
\end{equation}
where $r_0-Ax = b-A(x_0+\delta x)$ is the residual of $x=x_0+\delta x$. The equivalence 
\eqnref{gen_minimal_residual:eq} holds because 
the minimizer $\delta x$ is such that $A\delta x$ represents the orthogonal projection of 
$b-A(x_0+\delta x)$ onto $A\mathcal{V}$. So the minimal residual condition for GMRES or GCR is  
equivalent to 
\begin{equation} \label{rperp:eq}
   r_j \perp A K_j.
\end{equation}
GMRES works by building a sequence of orthonormal vectors  $v_1,\ldots,v_j$ which span
$K_j(A,r_0)$. 
The coefficients $\eta_i$ in $x_j = x_0 + 
\sum_{i=1}^j \eta_i v_i$ are then obtained by solving a small $(j+1) \times j$ least squares problem. 
GCR computes a sequence of orthonormal vectors $w_1,\ldots,w_j$ which span $AK_j$, 
and an auxiliary sequence $s_1,\ldots, s_j$ with $w_j = As_j$. The iterate $x_{j}$
is obtained by an update $x_j = x_{j-1} + \alpha_j s_{j-1}$. 
Both, GMRES and GCR, suffer from the fact that the number of vectors stored and the cost for the orthogonalizations 
increase linearly with $n$. A remedy is to restart or truncate the methods.
GCR requires twice as many vectors to store as GMRES, and its total arithmetic costs 
are slightly higher than in GMRES. 
If convergence is slow, GMRES is significantly more stable numerically 
than GCR as was proven in \cite{jiranek:1483}. There is one advantage of 
GCR over GMRES, however: While the GMRES algorithm has to be modified (slightly) to 
allow for variable preconditioning, GCR adapts itself ``automatically'' to such a situation. 
For this reason GCR was used in various contexts of variable 
preconditioners like, for example, \cite{Luescher2003}. 
In the present paper
we consider situations where there is an additional need for (inexact) deflation
in order to get sufficiently fast convergence. While this can be done elegantly and
efficiently within the GMRES framework, there seems to be no such way for GCR. This explains why 
we focus on GMRES from now on.

\subsection{Flexible GMRES} \label{fgmres:subsec}

Standard (non-variable) right preconditioning of the equation $Ax = b$ means that instead of
looking for an iterate $x_j \in x_0 + K_j(A,r_0)$ we now look for one in the space
$x_0 + MK_j(AM,r_0)$, where $M$ stands for the preconditioning matrix. The idea is that with
a good preconditioner $M$ (a good approximation to $A^{-1}$), the ``search space'' $x_0+MK_j(AM,r_0)$ 
contains substantially better approximations to the solution than $x_0+K_j(A,r_0)$. 
In right preconditioned GMRES \cite{saad_book} we obtain the $j$-th iterate $x_j$ as
\begin{equation} \label{right_prec:eq}
x_j = x_0 + M\delta y_j, \mbox{ where } \delta y_j = \argmin_{\delta y \in K_j(AM,r_0)} 
\| b-A(x_0+M\delta y)\|.
\end{equation}

Right preconditioned GMRES uses the Arnoldi process to compute an orthonormal basis 
of $K_j(AM,r^0)$, orthogonalizing $AMv^{j}$ against all previous vectors $v_1,\ldots,v_{j}$. 
The solution of \eqnref{right_prec:eq} can then be obtained as the solution of a small 
$(m+1) \times m$ least squares problem. We do not discuss this further here, since right 
preconditioned GMRES appears as a special case of flexible GMRES \cite{Saad:FGMRES} which we  
develop in all its details now. 

In a flexible context, the preconditioning matrix depends on the iterative step,
i.e.\ we have a new preconditioning matrix $M_j$ in each step $j$. In a ``flexible''
Arnoldi process we thus orthogonalize $AM_jv_{j}$ against all previous vectors $v_1,
\ldots, v_j$. It will turn out
useful to also keep track of the preconditioned vectors $z_j = M_j v_j$ for 
future use when computing (flexible) GMRES iterates. So we end up with 
the flexible Arnoldi process described as Algorithm~\ref{flexarno:alg}.  

\begin{algorithm}
    \LinesNumbered
    \DontPrintSemicolon
  \SetKwInOut{Input}{Input}
  \SetKwInOut{Output}{Output}
  \caption{Flexible Arnoldi Process\label{flexarno:alg}}
  \Input{$A\in \cnn$, $b \in \cn$, an integer $m$}
  \Output{$V_{m+1} = [v_1 | \ldots | v_{m+1}] \in \mathC^{n \x (m+1)}$, $Z_m=[z_1 | \ldots | z_m] \in \mathC^{n \x m}$  ,\\ $\widehat{H}_m=(h_{i,j}) \in \mathC^{{m+1} \x m}$}
  \BlankLine
  $\beta  =\|b\|$\;
  $v_1:= \frac{1}{\beta}b$\;
  \For{$j=1,\ldots, m$}{
      $z_j:=M_{j}v_j$ \tcc*[f]{preconditioning of new basis vector}  \;
      $w:= Az_j$ \;
      \For(\tcc*[f]{orthogon.\ against previous vectors}){$i=1,\ldots,j$}{
      	$h_{ij}:= v_i^H w$ \;
      	$w =w-h_{ij} v_i$ \;
        }
      $h_{j+1, j}:=\|w\|$ \;
      $v_{j+1}:= w /h_{j+1,j}$ \tcc*[f]{normalization}\;
  }
\end{algorithm}
 
The flexible Arnoldi relation
\begin{equation} \label{arnoldilike:eq}
	AZ_m = V_{m+1} \widehat{H}_m 
\end{equation}
summarizes the computations of Algorithm~\ref{flexarno:alg}, its $j$-th column representing  
\[
h_{j+1,j}v_{j+1} = Az_j - \sum_{i=1}^{j} h_{ij} v_i.
\]
Note that the $(m+1)\times m$ matrix $\widehat{H}_m$ is upper Hessenberg, i.e.\ 
its elements below the first subdiagonal, $h_{ij}$ for $i > j+1$, are all 0. The 
columns $v_j$ of $V_{m+1}$ are orthonormal, so we have $V_{m+1}^HV_{m+1} = I \in 
\mathC^{(m+1) \times (m+1)}$. 

In the non-flexible case, i.e.\ $M_j = M$ for all $j$, 
the orthonormal vectors $v_1,\ldots,v_j$ span 
the Krylov subspace $K_j(AM,r_0)$, and the vectors $z_i = Mv_i, i=1,\ldots,j$ span the search space 
$MK_j(AM,r_0)$. In this case we can 
formulate \eqnref{arnoldilike:eq} without $Z_m$ to retrieve the standard Arnoldi relation
\begin{equation} \label{arnoldi:eq}
(AM) V_m= V_{m+1} \widehat{H}_m.
\end{equation}

In the case of a variable preconditioner, there is no immediate notion of a (preconditioned) Krylov subspace, 
but we can still use the space spanned by the vectors $z_i$, i.e.\ $\range(Z_m)$ as a search space.
This is the approach taken in flexible GMRES (F-GMRES) \cite{Saad:FGMRES}, where we obtain the iterates $x_m$ by 
imposing the minimal residual condition
for $x_m \in x_0 + \range(Z_m)$. So we have to compute 
\begin{equation} \label{GMRES_ls:eq}
  \eta = \argmin_{\xi \in \mathC^{m}} \| b-A(x_0 + Z_m\xi) \|.
\end{equation}
Using the flexible Arnoldi relation \eqnref{arnoldilike:eq} and the fact that $V_{m+1}$ has 
orthonormal columns, we have
\[
  \| b-A(x_0 + Z_m\xi) \| =  \| r_0 - V_{m+1} \widehat{H}_m \xi \| = 
   \| V_{m+1} \beta e_1 - V_{m+1} \widehat{H}_m \xi \| = \| \beta e_1 -\widehat{H}_m \xi \|.
\]
So solving the $(m+1) \times m$ least squares problem $\eta = \argmin_{\xi \in \mathC^{m}}
 \|\beta e_1 -\widehat{H}_m \xi \|$ gives the coefficient vector $\eta$ from which
 we retrieve the FGMRES iterate $x_m$ as
\[
    x_m = x_0 + Z_m \eta.
\]
%

We note in passing that the $(m+1) \times m$ least squares problem is usually solved via 
a QR-factorization of $\widehat{H}_m$. Due to the upper Hessenberg structure of $\widehat{H}_m$,
this can easily be done in an iterative manner by updating the QR-decomposition of $\widehat{H}_j$
to that of $\widehat{H}_{j+1}$ using one additional Givens rotation. Moreover, it is 
possible to very cheaply update the norm of the residual of the $j$-th GMRES iterate without
explicitly computing the iterate. This can be used to implement an early termination criterion
in cases where the $j$-th (F)GMRES iterate is already accurate enough for $j < m$. Details can 
be found in \cite{saad_book}, e.g.. See also section~\ref{residual:sec}.

\subsection{Restarts and Deflation}
If larger values of $m$ are required to obtain a sufficiently accurate approximation $x_m$,
full F-GMRES as described in the previous section is not feasible since we have to store the
$2m$ vectors $v_j$ and $z_j$ and the total cost for the orthogonalizations in the flexible Arnoldi process
grow like $\mathcal{O}(nm^2)$. It is therefore common practice to use {\em restarted} F-GMRES.
The integer $m$ is fixed to a moderate value. One then goes through several cycles. Each cycle
performs $m$ steps of the full F-GMRES algorithm, with its initial guess given by
the approximation to the solution computed in the previous cycle.

While restarts overcome problems with storage and computational complexity, they may severely
degrade the convergence behaviour, and it might even happen that the method stagnates at a large 
residual norm without achieving any further progress. For difficult problems, this phenomenon 
has been widely observed in the literature, and several cures have been proposed. One of the 
most appropriate ones is the deflated restart modification which we discuss now.

Upon a restart, all the information contained in the search space $\mathcal{K}$ 
built up so far is lost. In particular, in the new cycle we do not
keep the residuals orthogonal to $A\mathcal{K}$. The idea of deflated restarts is
to extract a low-dimensional subspace $\mathcal{U}$ of $\mathcal{K}$ which contains the 
most important information in order to speed-up convergence, and to at least maintain 
the residuals orthogonal to $A \mathcal{U}$. The new cycle can then be characterized 
as an augmenting Krylov subspace method with $\mathcal{U}$ as the augmenting space. The 
crucial point is that, as we will see in the next section, it is possible to construct this deflating  augmenting subspace in such a way that we do not need additional multiplications with $A$.

\subsection{Augmentation by Deflation}
We first describe the ideas and technicalities of deflated restarts for the 
non-flexible case. The extension to FGMRES will be given in section~\ref{defl_flex:sec}. 
So for now we assume we have a fixed preconditioner $M$ in every iteration, and to further 
simplifiy the notation we base our whole discussion on the non-preconditioned system $Ax = b$.
The (right) preconditioned case is obtained by merely replacing $A$ with $AM$ 
everywhere. 

The presence of eigenvector components corresponding to small eigenvalues in a residual
$r$ induces a large contribution of these eigenvectors in the error $e = A^{-1}r$. 
It is thus desirable to chose the augmenting subspace $\mathcal{U}$ as one that contains
approximations to eigenvectors belonging to small eigenvalues. Those can be approximated from
the Krylov subspace built up in the previous cycle. If this is done in the right manner, 
one can establish an Arnoldi type process for the augmented spaces with little additional 
cost as compared to standard Krylov subspaces, so that 
one ends up with an efficient method. We now give the details, following \cite{morgan-dr},
which in turn builds on \cite{ChapmanSaad,Morgan:GMRESDR,Morgan:GMRES-IR}.
The key ingredient is to use harmonic Ritz vectors \cite{PaPaVdVo:EigValBounds}.

\begin{defn}\label{def:ritz}
Given a matrix $A \in \cnn$ and a subspace $\mathcal{K} \subset \cn$, a {\em harmonic Ritz pair} 
$(\lambda,v)$ with 
$\lambda \in \mathC, v\in \mathcal{K}$ for $A$ with respect to $\mathcal{K}$  satisfies 
\begin{eqnarray} \label{harmRitz:eq}
Av - \lambda v &\bot& A \mathcal{K}.
\end{eqnarray}
\end{defn}

Harmonic Ritz pairs are well suited to approximate eigenpairs with a small eigenvalue 
\cite{PaPaVdVo:EigValBounds}. The following lemma shows how harmonic Ritz pairs can be computed 
in cases where $\mathcal{K}$ is a Krylov subspace or an appropriately augmented Krylov 
subspace. 

\begin{lemma} \label{harm_Ritz:lem}
Let $\mathcal{K}$ be a subspace of $\mathC^n$ of dimension $m$ and assume
that 
\[
A \mathcal{K} \subset \mathcal{K} \oplus \langle s \rangle \mbox{ for some vector } s \in \mathC^n.
\]
Let the columns $v_1,\ldots,v_{m+1}$ of $V_{m+1} \in \mathC^{n \times (m+1)}$ be
orthonormal vectors such that $v_{1}, \ldots, v_m$ span $\mathcal{K}$ and 
$v_1,\ldots, v_{m+1}$ span $\mathcal{K} \oplus \langle s \rangle$, so that there
exists a full rank matrix $\widehat{H}_{m} \in \mathC^{ (m+1) \x m}$  with
\begin{equation} \label{Arnoldi_Ritz:eq}
AV_m= V_{m+1} \widehat{H}_{m}.
\end{equation}
Moreover, let
\[
\widehat{H}_m = \left[ \begin{array}{c} H_m \\ h_m^H \end{array} \right]
\enspace \mbox{ with } H_m \in \mathC^{m \times m}, \, h_m \in \mathC^m.
\]
Then:
\begin{itemize}
\item[a)] The harmonic Ritz pairs of $A$ w.r.t.\ $\mathcal{K}$ are given 
as $(\lambda, V_m g)$, where $(\lambda,g)$ is an eigenpair
of the matrix 
\[
   H_m + f_m h_m^H,  \mbox{ where }  H_m^H f_m = h_m.
\]
\item[b)] The residual $Au - \lambda u$ of a harmonic Ritz pair satisfies
\[
    Au - \lambda u \in \langle r \rangle,
\]
where $r$ spans the orthogonal complement of $A\mathcal{K}$ in $\mathcal{K} \oplus \langle s \rangle$,
i.e.\
\[
A\mathcal{K} \oplus_\perp \langle r \rangle = \mathcal{K} 
\oplus \langle s \rangle.
\] 
\end{itemize}  
\end{lemma}

\begin{proof}
Using \eqnref{Arnoldi_Ritz:eq}, the defining property \eqnref{harmRitz:eq} for the Ritz 
pair $(\lambda,v)$ with $v = V_m g$ turns into
\begin{eqnarray*}
& & (AV_m)^H(Av-\lambda v) = 0 \\
\Leftrightarrow & & (V_{m+1}\widehat{H}_m)^H(V_{m+1}\widehat{H}_m g - \lambda V_m g) = 0 \\ 
\Leftrightarrow & & \widehat{H}_m^H \left( \widehat{H}_m g -  \left[ \begin{array}{c} \lambda I_m \\ 0
    \end{array} \right] g \right) = 0.
\end{eqnarray*} 
Since $\widehat{H}_m$ has full rank, the matrix ${H}_m$ is non-singular. The last equality can thus be
further reduced to
\begin{eqnarray*}
& & \left(H_m^HH_m + h_m h_m^H \right) g - \lambda H_m^H g = 0 \\
\Leftrightarrow && \left(H_m + f_m h_m^H \right) g - \lambda  g = 0,  \mbox{ where } 
 H_m^H f_m = h_m,
\end{eqnarray*}
which shows a). Part b) is trivial, since the residual of any harmonic Ritz pair is in
$A\mathcal{K}+ \langle s \rangle$ and is orthogonal to $A\mathcal{K}$.
\end{proof}

So, to obtain harmonic Ritz pairs, we first compute $f_m$ as the solution of an $m\times m$ 
linear system and then compute the eigenpairs of $H_m + f_m h_m^H$. 

Note that by \eqnref{gen_minimal_residual:eq} the vector $r$ from part b) of Lemma~\ref{harm_Ritz:lem}
is also the residual of the GMRES iterate with respect to the search space $\mathcal{K}$. Since 
in the situation of Lemma~\ref{harm_Ritz:lem} we have 
\begin{equation} \label{harm_U:eq}
A \mathcal{U} \in \mathcal{U} + \langle r \rangle
\end{equation}
for any subspace $\mathcal{U}$ spanned by some harmonic Ritz vectors, we see that the augmented 
Krylov subspaces $K_j(A,r,\mathcal{U}) := \mathcal{U} + K_j(A,r)$ satisfy
\begin{equation} \label{harm-nested:eq} \nonumber
A \cdot ( \mathcal{U} + K_j(A,r) ) \subset \mathcal{U} + K_{j+1}(A,r).
\end{equation}
This has two major consequences. The first is that, starting from an orthonormal basis of 
$K_1(A,r,\mathcal{U})$, we can build nested orthonormal bases for the spaces $K_j(A,r,\mathcal{U})$
in an Arnoldi type manner. The second is that we can 
compute harmonic Ritz vectors with respect to $K_m(A,r,\mathcal{U})$  in the way given in 
Lemma~\ref{harm_Ritz:lem}. 

This is why we are able to set up a restarted GMRES method with deflation of harmonic Ritz vectors,
where each cycle consists of the following steps: 

\begin{enumerate}
  \item Extract $k$ harmonic Ritz pairs from the ``search space'' $\mathcal{U}^{pr}+K_m(A,r^{pr})$ 
        built in the previous cycle; 
        these harmonic Ritz vectors span the current augmenting subspace $\mathcal{U}$.
  \item With $r$ denoting the residual of the iterate $x$ at the beginning of the current cycle,
        compute nested orthonormal bases $v_1,\ldots,v_{k+j}$ for $K_j(A,r,\mathcal{U}), j=1,\ldots,m,$
        such that for $V_{k+j}= [v_1|\ldots|v_{k+j}]$ we have the Arnoldi type relation
        \begin{equation} \label{Arnoldi-type:eq}
             AV_{k+j} = V_{k+j+1} \widehat{H}_{k+j}, \, \widehat{H}_{k+j}\in \mathC^{(k+j+1) \times (k+j)}.
        \end{equation}
 \item  Compute the current GMRES iterate and its residual using \eqnref{Arnoldi-type:eq}.
\end{enumerate}

Here, as in the sequel, we use the superscript $pr$ to denote quantities from the previous cycle; 
no superscript refers to the current cycle. For an efficient implementation, all three steps 
are coupled by the way in which we construct the orthogonal basis for the current search space 
$K_j(A,r,\mathcal{U})$. We now discuss the three steps in detail.

\paragraph{Extracting the Ritz pairs}
Let $V^{pr}_{\widetilde{m}}, \widetilde{m}=k+m,$ denote the matrix containing the orthonormal basis vectors 
of the search space built up in the previous cycle. By Lemma~\ref{harm_Ritz:lem}a),
the $k$ harmonic Ritz vectors $u_\ell, \ell = 1,\ldots,k$ to augment the current Krylov subspace 
can be computed using $\widehat{H}_{\widetilde{m}}^{pr}$ from the Arnoldi-type relation \eqnref{Arnoldi-type:eq}
of the previous cycle. The harmonic Ritz vectors are then given as 
\[
u_\ell = V_{\widetilde{m}}^{pr} g_\ell, \; g_\ell \in \mathC^{\widetilde{m}}, \, \ell = 1,\ldots,k,
\]
which we summarize as
\begin{equation} \label{harm_G:eq}
   U_k := [u_1|\ldots|u_k] = V_{\widetilde{m}}^{pr} G_k, \enspace G_k = [g_1| \ldots |g_k].
\end{equation}

\paragraph{Computing the next GMRES iterate and its residual.} Assume that we have already built 
$V_{k+m+1}$, the columns which represent an orthonormal basis of the current search space.
We have to compute
$\eta \in \mathC^{k+m}$ such that $\| b-A(x+V_{k+m}\eta)\|$ is minimal. 
The residual $r=b-Ax$ of the current iterate, which is in $K_{1}(A,r,\mathcal{U}^{pr})$ and thus 
in $\range(V_{k+1})$,
can be expressed as 
\[
r = b-Ax = V_{k+m+1} c, \enspace \mbox{where } c =  \left[ \begin{array}{c} V_{k+1}^H r
 \\  0 \end{array} \right].
\]
Taking \eqnref{Arnoldi-type:eq} as granted for the moment, we have
\[
  b-A(x+V_{m+k}\eta) =  V_{m+k+1}\left( c - \widehat{H}_{m+k} \eta \right).
\]
Since the columns of $V_{m+k+1}$ are orthonormal, minimizing $b-A(x+V_{m+k}\eta)$ is
therefore equivalent to solving the small $(m+k+1) \times (m+k)$ least squares problem
\[
  \eta = \argmin  \left\| c
-\widehat{H}_{m+k} \eta  \right\|.
\]
This gives the next GMRES iterate $x^{next} = x + V_{m+k} \eta$ with residual
\begin{equation} \label{r:eq}
r^{next} = V_{m+k+1}(c- \widehat{H}_{m+k} \eta).
\end{equation}

\paragraph{Building the orthonormal basis for the current search space}
The crucial part here is to obtain an orthonormal 
basis for $K_1(A,r,\mathcal{U})$ without investing multiplications with $A$.

From the columns of $V_{\widetilde{m}}^{pr}$ being orthonormal, we obtain an orthonormal basis 
$v_1,\ldots,v_
k$ of $\mathcal{U}$ by computing an orthonormal
basis for $\range G_k$ from \eqnref{harm_G:eq}, i.e.\ by computing the 
QR-factorization $G_k = Q_k R_k$, where $Q_k \in \mathC^{\widetilde{m} \times k}$ 
has orthonormal columns, and putting
\[
{V}_k := [{v}_1 | \ldots | {v}_k]  =  V_{\widetilde{m}}^{pr} Q_k.
\] 
From \eqnref{rperp:eq} we have $A\mathcal{U} \subset \mathcal{U} + \langle r \rangle$, 
where $r$ is the residual of the iterate at the beginning of the current cycle, 
$r = V^{pr}_{\widetilde{m}+1} \left( c^{pr} - \widehat{H}_{\widetilde{m}}^{pr} \eta^{pr} \right)$ according to \eqnref{r:eq}.
Defining $G_{k+1}\in \mathC^{(\widetilde{m}+1) \x (k+1)}$ as
\begin{equation}  \label{G_kplus1:eq}
G _{k+1} = \left[ \begin{matrix}
	\begin{matrix} G_k \\ 0\end{matrix} & c^{pr} - \widehat{H}_{\widetilde{m}}^{pr}\eta^{pr}
\end{matrix} \right],
\end{equation}
we thus have that $\range(V^{pr}_{\widetilde{m}+1} G_{k+1}) = \mathcal{U} + \langle r \rangle$.
So we can extend $v_1,\ldots,v_k$ to an orthonormal basis $v_1,\ldots,v_{k+1}$ of
$\mathcal{U} + \langle r \rangle$ by extending the QR-factorization $G_k = Q_k R_k$ 
to a QR factorization $G_{k+1} = Q_{k+1} R_{k+1}$ where
\[
Q_{k+1} = \left[ \begin{array}{cc}
               Q_k & * \\
               0   & *
            \end{array}
          \right] \in \mathC^{(\widetilde{m}+1) \times (k+1)}. 
\]
This gives the relation
\begin{equation}\label{tildeV_Vplus1:eq}
{V}_{k+1} = [{v}_1| \ldots| {v}_{k} | {v}_{k+1} ] = V^{pr}_{\tilde{m}+1}Q_{k+1}.
\end{equation} 

On the other hand, by \eqnref{harm_U:eq} for $j=0$ there exists a matrix
$\widehat{H}_k  \in \mathC^{(k+1) \times k}$ such that 
\begin{equation} \label{tilde_Arnoldi:eq} 
AV_{k} = V_{k+1} \widehat{H}_{k}.
\end{equation} 
This matrix is actually given as  
\begin{equation} \label{hatHk:eq}
\widehat{H}_k = Q_{k+1}^H \widehat{H}_{\widetilde{m}}^{pr} Q_k,
\end{equation}
which can be seen by comparing
\begin{eqnarray*}
	A {V}_{k} &=& A V_{\widetilde{m}}^{pr}Q_k \, = \, V_{\widetilde{m}+1}^{pr} \widehat{H}_{\widetilde{m}}^{pr} Q_{k} 
\end{eqnarray*}
which follows from the Arnoldi relation of the previous cycle and 
\[
 A {V}_{k} \, = \, {V}_{k+1} \widehat{H}_{k} \, = \, 
V_{\widetilde{m}+1}^{pr} Q_{k+1} \widehat{H}_{k},
\]
which is due to \eqnref{tildeV_Vplus1:eq} and \eqnref{tilde_Arnoldi:eq}. We have therefore shown \eqnref{Arnoldi-type:eq} for $j=0$. 

The Arnoldi-type relation \eqnref{tilde_Arnoldi:eq} for an orthonormal basis
of $\mathcal{U} + K_1(A,r)$ can now be extended to an orthonormal basis of 
$\mathcal{U} + K_j(A,r)$ through the standard Arnoldi orthogonalization procedure, where  
in step $j$ we orthogonalize $Av_{k+j}$ against all previous basis vectors. As a consequence,
the matrices $\widehat{H}_{k+j}$ in the resulting Arnoldi-type relation 
\eqnref{Arnoldi-type:eq} have the form
\[
\widehat{H}_{k+j} = \left[ \; \; \begin{matrix}
                             \begin{matrix}\widehat{H}_k  
                             \end{matrix} \hfill \vline & \begin{matrix} h_{1,k+1} & \ldots & \ldots & h_{1,k+j} \\
                                             \vdots    &  \vdots & \vdots & \vdots \\
                                             h_{k+1,k+1} & \ldots & \ldots & h_{k+1,k+j} 
                             \end{matrix} \\ \cline{1-2} 
                             0_{j\times k}
                               \hfill \vline   & \begin{matrix} h_{k+2,k+1} & \ldots & \ldots & h_{k+2,k+j}\\
                                               0        & \ddots & \vdots      & h_{k+3,k+j} \\
                                               \vdots   & \ddots  & \ddots     & \vdots \\ 
                                               0        & \ldots  & 0          & h_{k+j+1,k+j}
                                            \end{matrix}
                            \end{matrix} \; \;
                        \right] \in \mathC^{(k+j+1) \times(k+j)},
\]
which has upper Hessenberg structure except for the first $k+1$ rows. Herein, $\widehat{H}_k$ is from
\eqnref{hatHk:eq}; for the other entries $h_{i\ell}$ see Algorithm~\ref{fgmresdr:alg} below.

\subsection{Deflating flexible GMRES} \label{defl_flex:sec}

If we use a constant preconditioner $M$, the deflating procedure described so far can be directly extended to right preconditioned GMRES by just changing the operator from $A$ to $AM$.
Note that we now compute harmonic Ritz values and vectors for $AM$ rather than $A$. 

As was pointed out in \cite{pinel:phd,GirGrattPinVas:FGMRESDR}, going from constant to varying preconditioning, 
is conceptionally almost trivial: 
For each $j$ the preconditioner $M_j$ occurs only once in the matrix-vector product
$z_j = M_j v_j$ of the flexible Arnoldi process. Since the vectors $v_j$ are linearly independent---they are even mutually orthogonal---there exist a constant matrix $M$ for which
$Mv_j = M_jv_j$ for $j=1,\ldots,m$. We don't have to know $M$ explicitly, all we need is its action on the $v_j$ which is given through the variable preconditioning. So we can perform the deflation procedure in exactly the same way as in the case of a constant preconditioner.   

Algorithmically, we have to take care of storing the preconditioned vectors $z_j = M_jv_j$ explicitly, as it is done in the flexible Arnoldi process.
The resulting deflated, flexible restarted GMRES method (FGMRES-DR) is formulated as
Algorithm \ref{fgmresdr:alg}.

\begin{algorithm}
    \LinesNumbered
    \DontPrintSemicolon
  \SetKwInOut{Input}{Input}
  \SetKwInOut{Output}{Output}
  \caption{F-GMRES-DR } \label{fgmresdr:alg} 
  \Input{$A\in \cnn$, $b \in \cn$, restart length $m$, dimension of augmenting (harmonic Ritz)
    subspace $k$}
  
  \BlankLine
  choose initial guess $x_{0}$, compute $r_0 = b-Ax^0$\;
  \tcc*[f]{first cycle: flexible GMRES, no deflating subspace}\;
  Perform $m$ steps of the flexible Arnoldi process starting with $r_0$. 
  This gives $Z_m, V_{m+1}, \widehat{H}_m$ from \eqnref{arnoldilike:eq}.\;
  Put $c = \|r_0\| e_1 \in \mathC^{m+1}$. \;  
  Compute  $\eta = \argmin_{\xi \in \mathC^{m}} \| c - \widehat{H}_{m} \xi\|$ \;
  Obtain GMRES iterate $x = x_0 + Z_m \eta$ and residual $r = r_0 - V_{m+1}\widehat{H}_m\eta$\;
   \Repeat(\tcc*[f]{all other cycles}){convergence}{   
        \tcc*[f]{get augmenting (harmonic Ritz) subspace }\;
	Compute all ${m}$ eigenpairs $(\lambda_i,g_i)$ of 
        $H_{m} + f_{m} h^H_{m}$ \tcc*[f]{Lemma~\ref{harm_Ritz:lem}a)} 
        \;
        Identify the $k$ smallest (in modulus) harmonic Ritz values $\lambda_i$, 
        collect the corresponding Ritz vectors $g_i$ as columns of $G_k 
        \in \mathC^{m \times k}$. \;
        \tcc*[f]{initialization for augmented Arnoldi process}   \;
        Build $G_{k+1} = \left[ \begin{matrix}
	\begin{matrix} G_k \\ 0\end{matrix} & c - \widehat{H}_{m}\eta
\end{matrix} \right]$ from \eqnref{G_kplus1:eq} and compute its QR-fac\-tori\-za\-tion
        $G_{k+1} = Q_{k+1} R_{k+1}$. \;
        Update $\widehat{H}_k := Q_{k+1}^H \widehat{H}_{m} Q_k$, where $Q_k$ results 
        from $Q_{k+1}$ by deleting its last row and column. \;
        Update $V_{k+1} = V_{m+1}Q_{k+1}, Z_{k} = Z_{m}Q_{k}$. \;
                \tcc*[f]{Remaining part of augm.\  flex. Arn.} \;
        \For{$j=k+1,\ldots, m$}{
            $z_j:=M_j v_j$   \;
            $w:= Az_j$ \;
            \For{$i=1,\ldots,j$}{
      	        $h_{ij}:= v_i^Hw$ \;
      	        $w =w-h_{ij} v_i$ \;
            }
           $h_{j+1, j}:=\|w\|$ \;
           $v_{j+1}:= w /h_{j+1,j}$ \;
        }
        \tcc*[f]{express (old) residual in terms of new basis} \;
        Compute $c = \left[ \begin{array}{c} V_{k+1}^H r
 \\  0 \end{array} \right] \in \mathC^{m+1}$ 
        \tcc*[f]{so $r=V_{m+1} c$} \;
        \tcc*[f]{Now obtain iterate and residual for this cycle} \;
        Compute  $\eta = \argmin_{\xi \in \mathC^{m}} \| c - \widehat{H}_{m} \xi\|$ \;
        Update GMRES iterate $x = x + Z_{m} \eta$ and corresponding residual $r = r - V_{m+1}\widehat{H}_{m}\eta$\; 
	}
\end{algorithm}

\subsection{Computing the implicit norm of the residual within GMRES-DR} \label{residual:sec}

As noted in section \ref{fgmres:subsec}, the QR factorization of $H$ can be updated at each step using Givens rotations. 
The \emph{implicit} norm of the residual is easily computed by applying in sequence the rotations to the right hand side $c$ of the little least squares problem and computing the absolute value of the last element of $c' = Q^Hc$. 
The set of rotations form the unitary matrix $Q^H$ such that $Q^H_{m+1}H_{m+1} = R_{m}$ with $R_m$ triangular. 
The application of $Q^H_{m+1}$ to $c$ is easily seen as a change of basis. 
This new basis is composed by a set of orthonormal vectors which spans the subspace inverted by the iteration, plus a vector which is not inverted. 
The first $m$ vectors of the new basis $\tilde{V}_{m+1}$ satisfy $AZ_m=\tilde{V}_mR$. The $m+1$ element of $c'$ is the component of $b$ on the non-inverted, normalized vector $\tilde{v}_{m+1}$.
In GMRES-DR the matrix $H$ is not Hessenberg. It is however possible to compute the implicit norm of the residual by computing a QR factorization of the non-Hesseberg submatrix of $H$ and then updating the QR factorization by Givens rotations.

\subsection{Mixed precision}

On most of modern CPUs, GPUs as well as on the CELL processor, single precision arithmetic is twice as fast as
double precision arithmetic. The natural way of exploiting the possible acceleration provided 
by single precision arithmetic units, while aiming at a full double precision result when solving a linear system, is by using the technique of \emph{iterative refinement}, see, e.g., \cite{Demmel06}.

Given two different available machine precisions, a high precision $p_h$ and a low precision $p_l$, iterative refinement 
is based on the algorithmic principle described in Algorithm~\ref{iter_refi:alg}.

\begin{algorithm}
    \LinesNumbered
    \DontPrintSemicolon
  \SetKwInOut{Input}{Input}
  \SetKwInOut{Output}{Output}
  \caption{Iterative Refinement} \label{iter_refi:alg} 
  
  \BlankLine
  choose initial guess $x_{0}$  \tcc*[f]{precision $p_h$} \;
  compute $r_0 = b-Ax^0$ with precision $p_h$\;
  \Repeat(\tcc*[f]{refinement cycles}){convergence}{   
  convert $r_0$ to precision $p_l$,   $r \leftarrow r_0$\;
  solve $Ax=r$ up to precision $p_l$ using $p_l$\;
  convert $x$ to precision $p_h$ and update $x_{0} \leftarrow x_{0} + x$\;
  compute $r_0 = b-Ax^0$ with precision $p_h$;\
  } 
\end{algorithm}

In practice, where we work with IEEE single ($p_l = 2^{-24} \approx 10^{-7}$) and double ($p_h = 2^{-53} \approx 10^{-16}$) precision and where we aim at a final norm of the residual of the order of $10^{-12}$, for example, two cycles are usually sufficient to obtain the desired accuracy. However, if the solver 
used in each refinement cycle is deflated flexible GMRES, we may encounter problems if we want to use the
harmonic Ritz vectors from the last GMRES cycle of the current refinement cycle as a deflating subspace
for the first GMRES cycle in the next refinement cycle: Upon convergence in low precision, the updated, low precision residual obtained via FGMRES-DR and the explicitly computed high precision residual may differ substantially. Then, including the more precise, explicitly computed high precision residual into the next deflated GMRES cycle is not possible, even when converted to low precision, because the fundamental relation \eqnref{harm_U:eq} between the residual and the harmonic Ritz vectors, which is at the heart of the efficient use of deflation in FGMRES-DR, is lost.      

To cope with this situation, we developed the following approach: Our refinement cycles correspond one-to-one to the
restart cycles of FGMRES-DR, so that we have more than just the 2 or 3 usual refinement cycles. 
Still, of course, the overwhelming part of all computation is done in low precision arithmetic.
 We then recompute the residual $r=b-Ax$ of the current iterate $x$ explicitly in high precision 
and convert it back to low precision. This explicitly computed residual $r$ is used in  Algorithm~\ref{fgmresdr:alg} 
to obtain the right hand side $c$ of the least squares problem that has to be solved in each cycle of 
FGMRES-DR (last 4 lines), and for $c - \widehat{H}_{m}\eta$ when obtaining the matrix $G_{k+1}$ at the beginning of the repeat-loop. We use $r$ also when updating the 
last vector of $V_{k+1}$. This last vector corresponds to the residual after orthonormalization against the updated $V_{k}$.
Our approach can thus be regarded as an attempt to sneak in exact residual information via the right hand side of the 
least squares problem and one vector of the subspace while at the same time maintaining the crucial 
relation \eqnref{harm_U:eq} although it will probably not hold exactly after using $r$ to get  $v_{k+1}$.
  
This approach works satisfactorily well in practice. Nevertheless, it may still happen that after some cycles the updated and the 
explicitly recomputed residual differ by more than a small amount. In this case, we perform a ``clean'' restart
of the method, i.e.\ we start Algorithm~\ref{fgmresdr:alg} from scratch with a first cycle that does not use any deflating subspace. 
As a criterion for triggering the clean restart we use $\|r^e\|/\|r^i\| > t$, with $r^e$ being the explictly 
recomputed residual, $r^i$ being the implicit residual computed as described in section \ref{residual:sec}, and $t$ being a 
tunable threshold.
From our experiments, $2\leq t \leq 10$ is a reasonable search range. 

\section{Numerical results} \label{example:sec}
We now report results of several numerical experiments which illustrate the impact of including deflation into the restarted flexible GMRES method. Our target systems are dynamically produced Wilson-Dirac operators with clover improvement \cite{Sheikholeslami:1985ij}. We start with the results of a series of experiments for a relatively small $16^3 \times 32$ lattice. We used a thermalized configuration obtained from a dynamical simulation
at $\beta = 5.29$ and $\kappa_{\rm sea} = 0.135$. The value of $c_{sw}$ in the clover improvement was $1.91$; the resulting critical value for the hopping parameter $\kappa$ was $\kappa_c \approx 0.13707$. 

Our implementation used 4 nodes of QPACE, i.e.\ 32 CELL SPU cores in total. For the domain decomposition, we divided the lattice into 64 sublattices of size $4^3 \times 8$ each, so that each of the 32 cores is assigned $8$ of these sublattices. We always performed 8 SAP-cycles for preconditioning, and the solves for the subdomains were done approximately via 5 steps of MR, independently of the subdomain and of the accuracy of the solution thus obtained. This choice of parameters was found to be the best by an extensive  numerical study.

Figure~\ref{1368:fig} shows the results for flexible GMRES without deflation (left) and with deflation (right) 
for various choices of the cycle length $m$. Here, as everywhere, we plot the 2-norm of the residual against
the number of matrix-vector multiplications invested, i.e.\ for every ``interior'' step of the restarted GMRES cycle. In these experiments the hopping parameter was set to $\kappa = 0.1368$, which is 
still quite far from $\kappa_c$, so that we may consider the system as relatively well conditioned. The figures illustrate the fact that a larger value of the cycle length $m$ results in fewer iterative steps. Deflation results in fewer iterations when the cycle length is fixed. Note that the arithmetic cost 
related to the Arnoldi process grows like $nm^2$ due to the orthogonalizations, so larger cycle lengths 
$m$ can significantly increase the computing time. This being said, Figure~\ref{1368:fig} shows that by 
investing just 4 vectors for deflation, we get a better performance by using a restart value $m=16$ than by 
using a restart value $m$ twice as large and no deflation.


\begin{figure}
\centerline{\includegraphics[width = 0.62\textwidth]{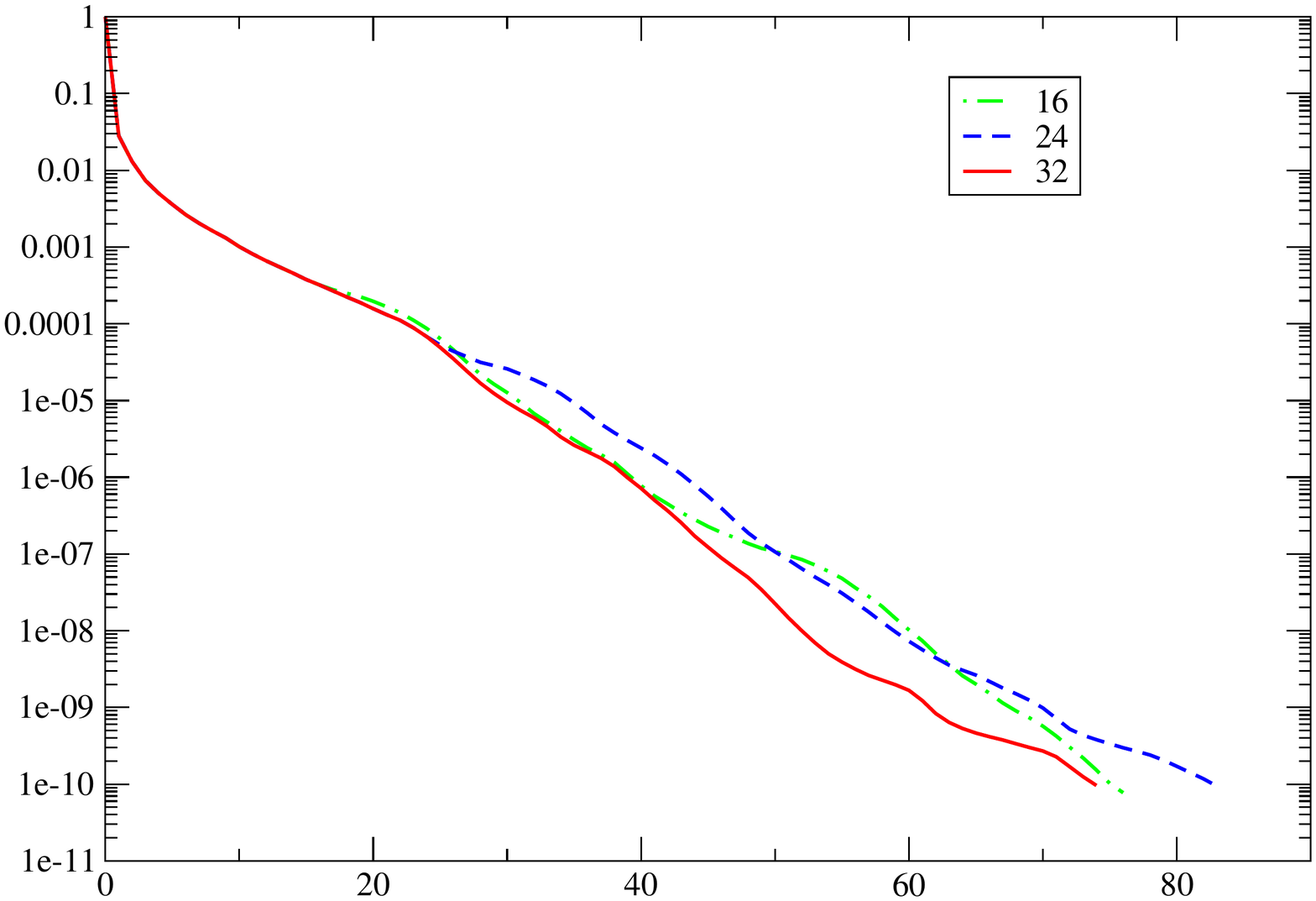} \hfill
\includegraphics[width = 0.62\textwidth]{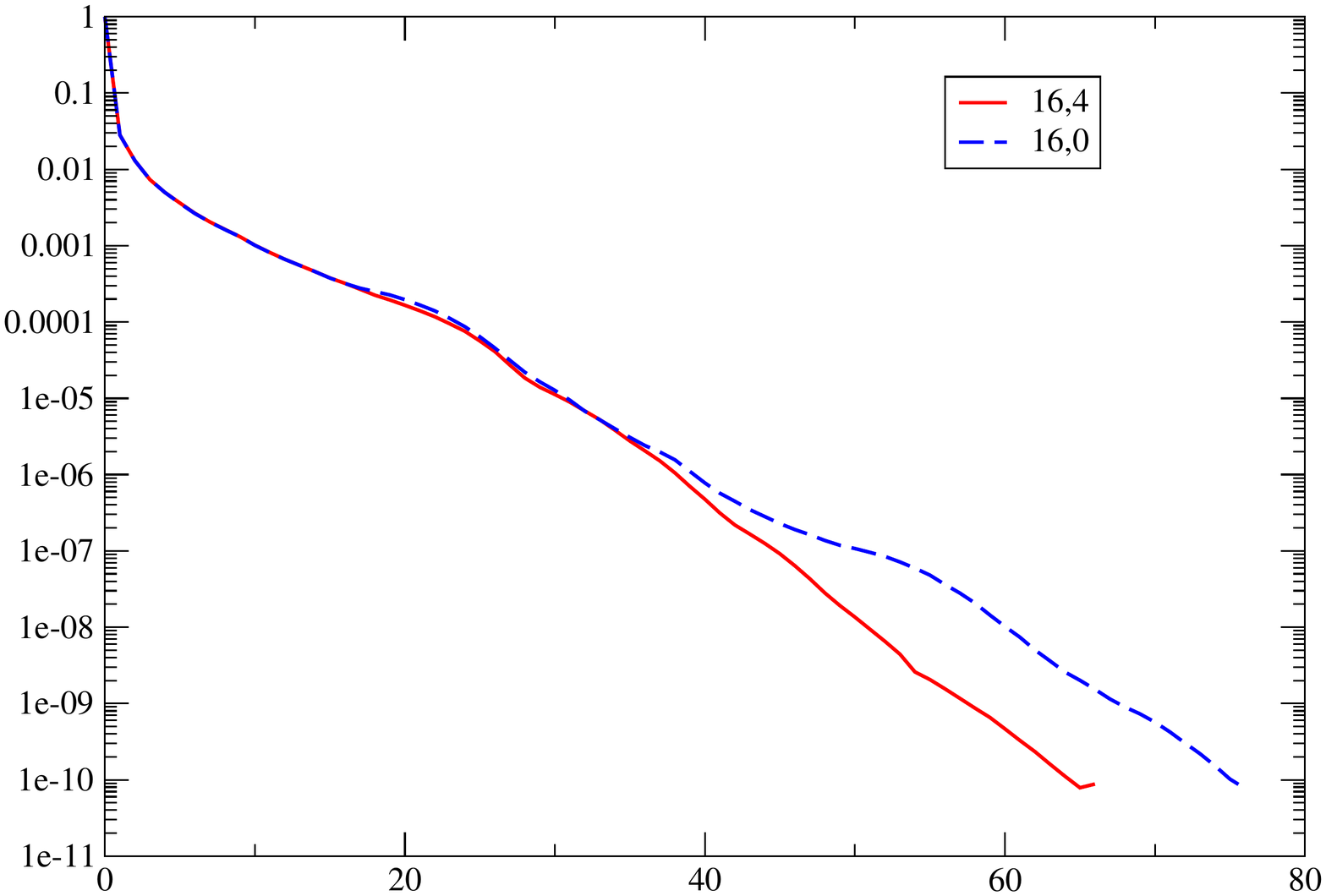}
}
\caption{Results without (left) and with (right) deflation. $16^3 \times 32$ configuration, $\kappa_c \approx 0.13707$,
$\kappa = 0.1368$. The legend on the left figure refers to the cycle length $m$; in the right figure the pair $(m,k)$ denotes the cycle length $m$ and the size of the deflating subspace $k$.\label{1368:fig}} 
\end{figure}

Figure~\ref{1371:fig} presents similar experiments for a different value of the hopping parameter, 
$\kappa = 0.1371$, which is very close to the critical value. Now the system is much more ill conditioned, and 
we see that deflation starts to have an even more significant impact on the performance of the method.
For example, reserving $k=4$ vectors for deflation while using a restart value of $m=16$ reduces the required 
iterations from about $450$ to $115$. The best deflated method uses $m=20$ and deflates $k=8$ vectors. It 
requires roughly only half as many iterations as the best non-deflated method which requires a much larger 
cycle length, $m=32$. 
In some of the residual plots in the figure to the right we observe an increase of the residual around iterations 
50-70 for some choices of the parameters. At these points we had to do a ``clean'' restart of FGMRES-DR, and 
the jump in the residual indicates the that the explicitly computed residual is considerably larger than the 
updated one. Since a clean restart discards all subspace information acquired so far, the subsequent GMRES cycle 
is comparably as slow as the very first cycle of the iteration. This shows that it might be advisable to use quite 
small values for the subspace dimension $k$, since in these cases a clean restart was never necessary, thus 
resulting in the fastest overall convergence. 

\begin{figure}
\centerline{\includegraphics[width = 0.62\textwidth]{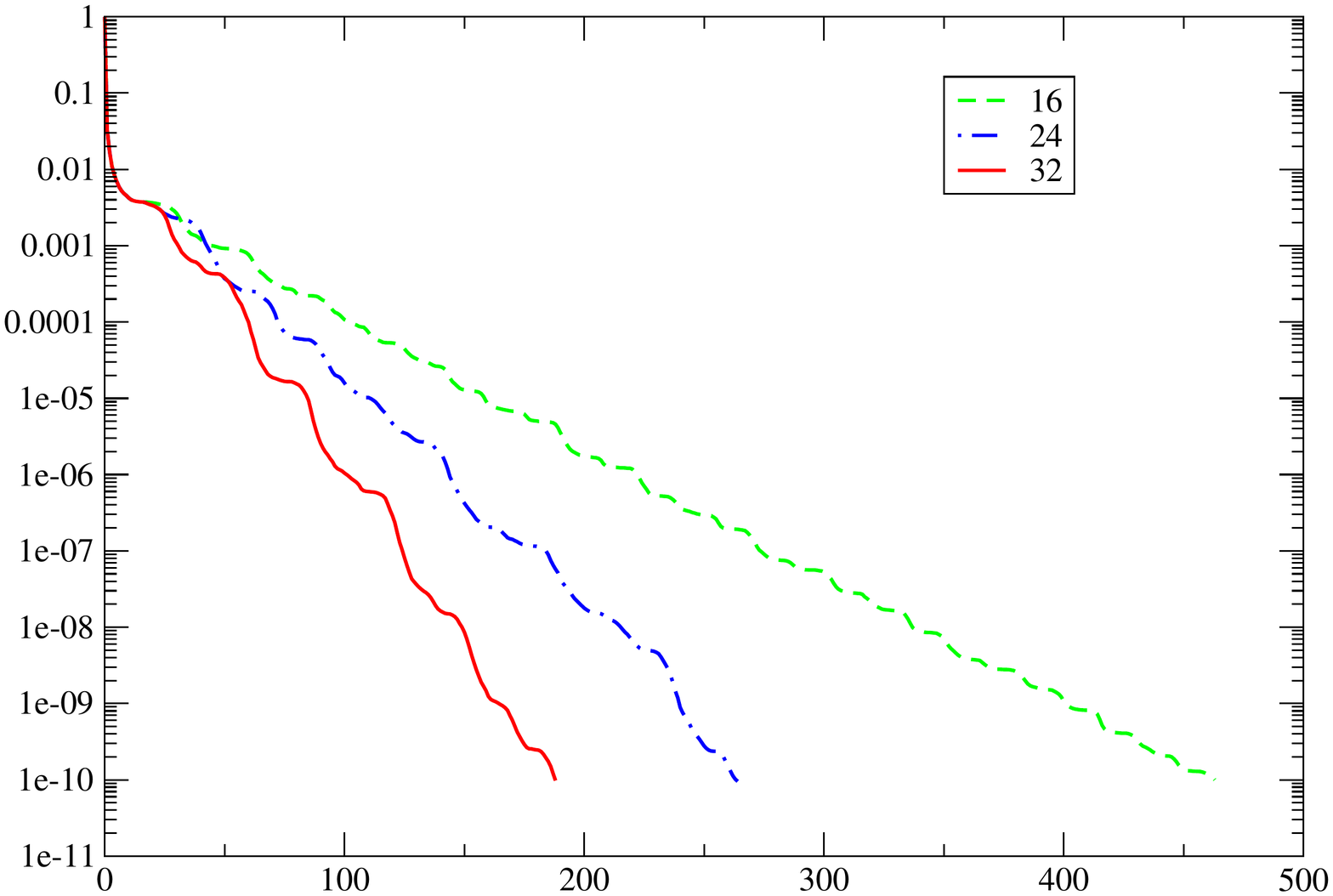} \hfill
\includegraphics[width = 0.62\textwidth]{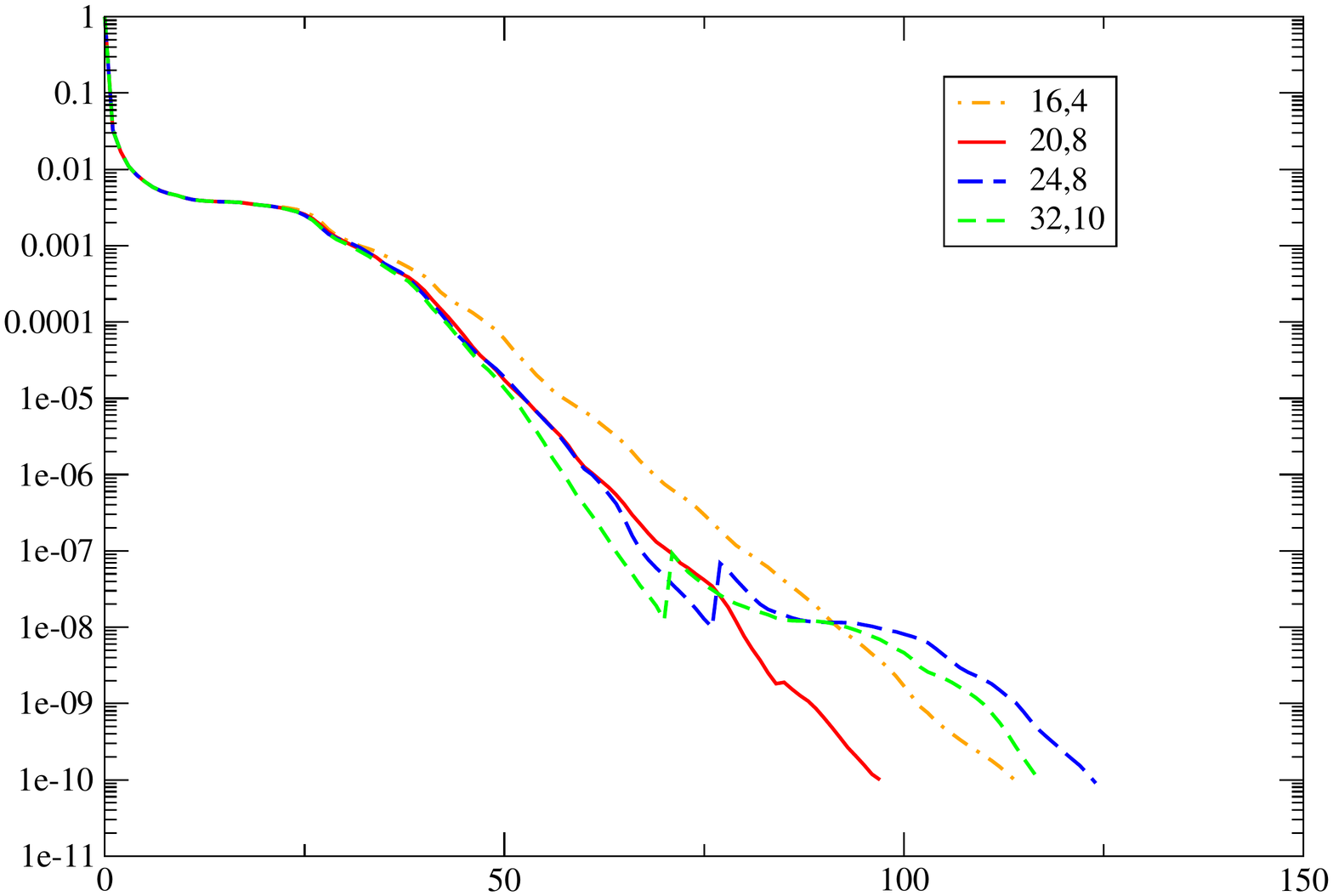}
}
\caption{Results without (left) and with (right) deflation. $16^3 \times 32$ configuration, $\kappa_c \approx 0.1371$,
$\kappa = 0.1371$. \label{1371:fig}} 
\end{figure}

\begin{figure}
\centerline{\includegraphics[width = 0.62\textwidth]{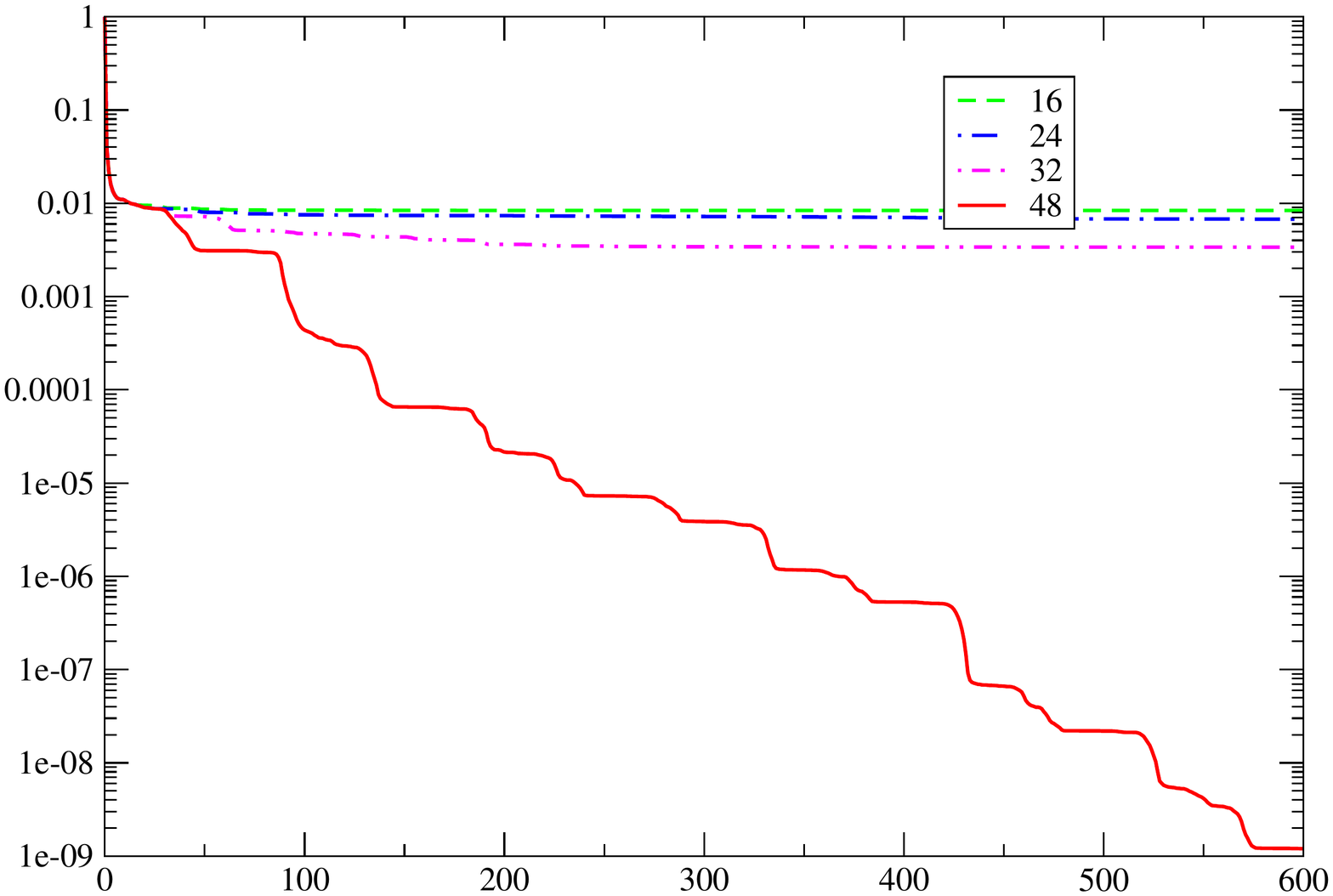} \hfill
\includegraphics[width = 0.62\textwidth]{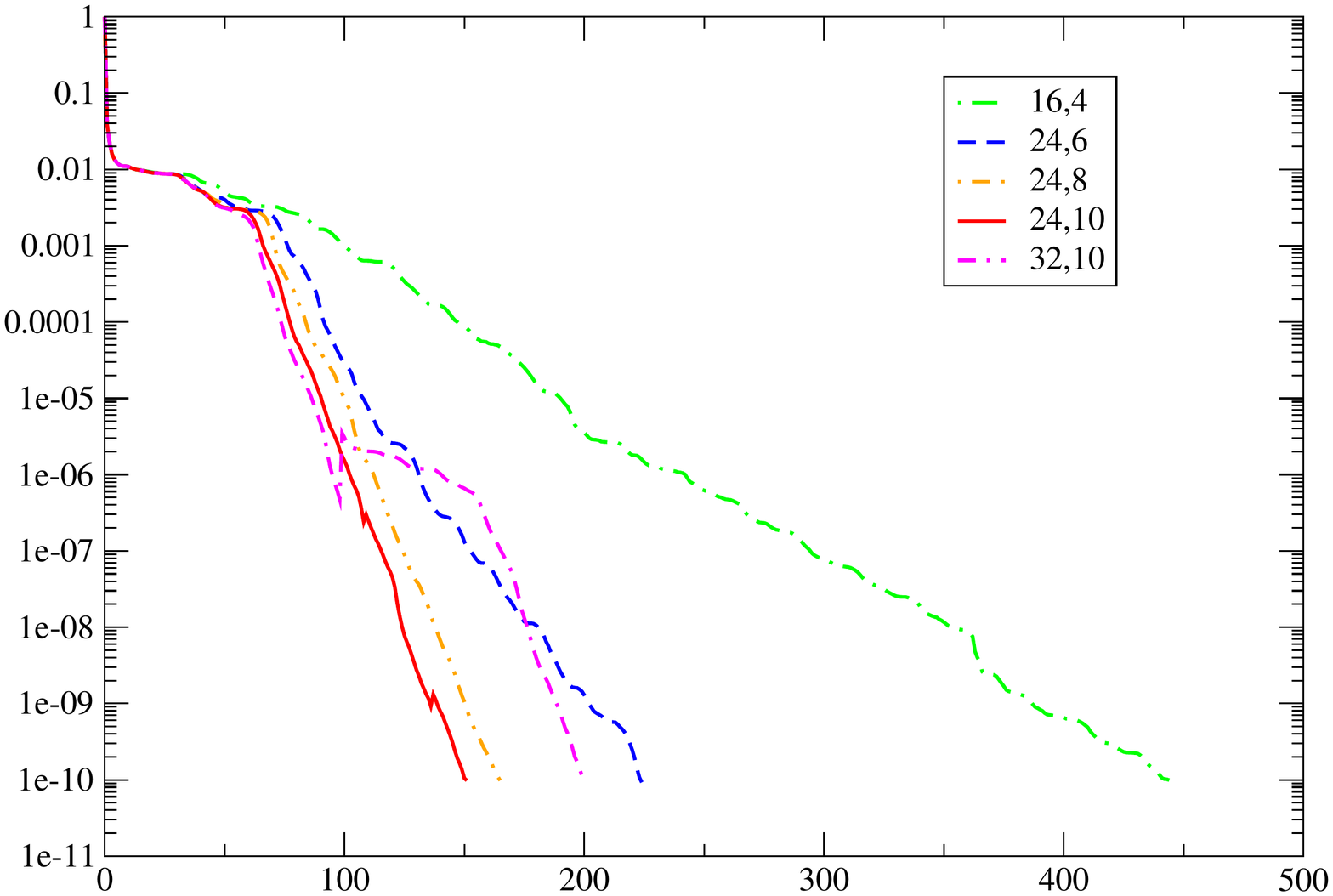}
}
\caption{Results without (left) and with (right) deflation. $16^3 \times 32$ configuration, $\kappa_c \approx 0.1371$,
$\kappa = 0.1374$.\label{1374:fig}} 
\end{figure}

\begin{figure}
\centerline{\includegraphics[width = 0.62\textwidth]{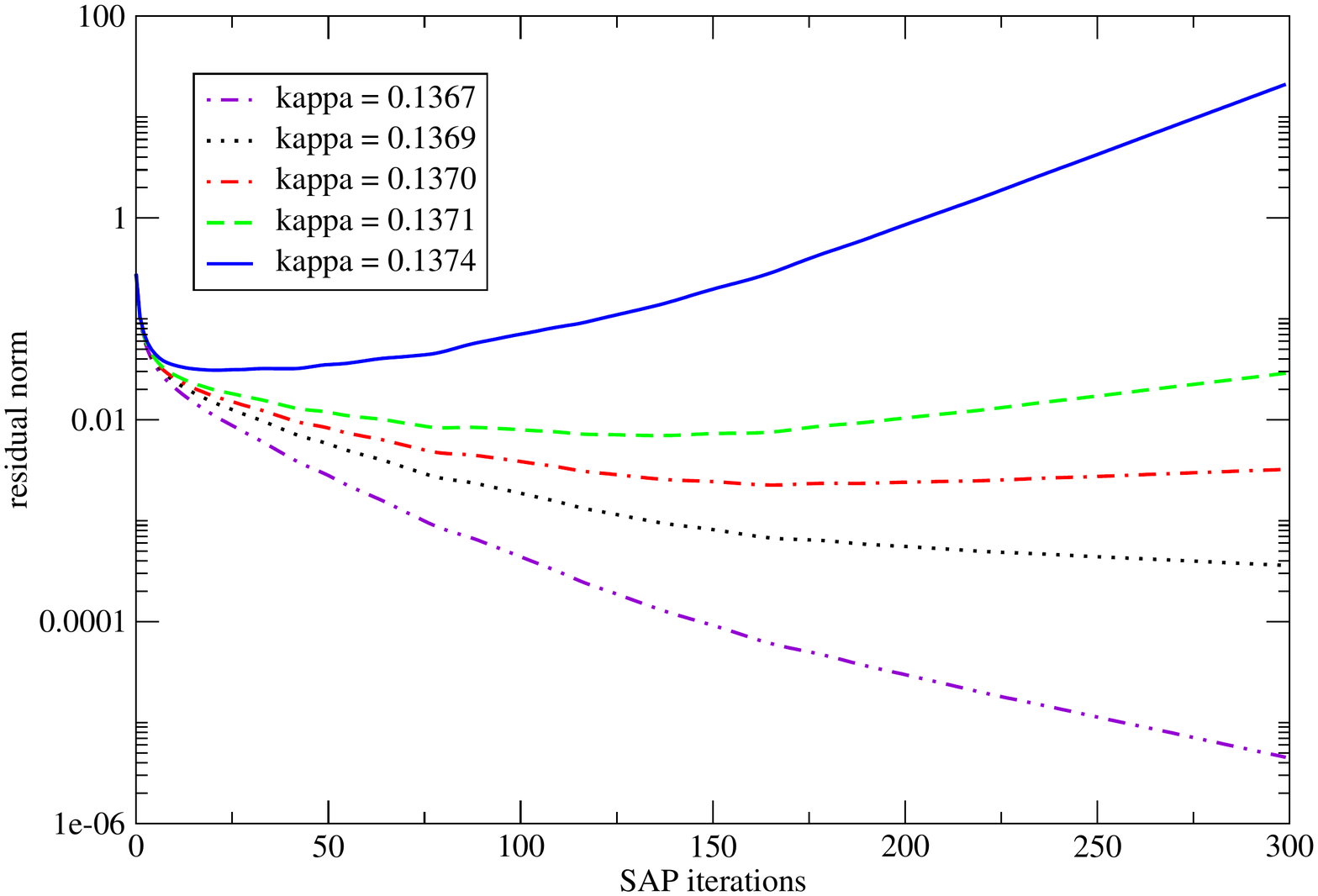}}
\caption{Convergence and divergence of the SAP method as a stand-alone solver for various values of $\kappa$\label{stand_alone:fig}}
\end{figure}

In a last experiment for this configuration we set $\kappa = 0.1374$, which is beyond the critical value. 
The SAP iteration by itself is divergent for this choice of $\kappa$ (see Figure~\ref{stand_alone:fig}), 
and when used as a preconditioner for GMRES we observe complete stagnation unless the restart value $m$ 
is large enough ($m=48$). On the other hand, deflation with a relatively small dimension of the 
deflating subspace results in quite fast convergence for moderate values of the restart value $m$. 
This indicates that the divergence of the pure SAP iteration is due to the fact that just a few eigenvectors belonging to the 
corresponding iteration matrix  are troublesome and that they are effectively removed by deflation. 

The second set of experiments consists of inversions of a state-of-the-art $48^3 \times 64$ $N_f=2$ lattice configuration thermalized at $\beta=5.40, \kappa_{\rm sea} = 0.1366$ and $c_{sw}=1.8228$, while  $\kappa_c \approx 0.13670$

In this case a full 256 nodes QPACE rack was used. The machine is configured as a 3D $4 \times 8 \times 8$ torus and we used an SAP block size of $8\times 2 \times 6 \times 6$ that fits the 8 SPU local stores  nicely and satisfies all the lattice and implementation constraints. 
For the test inversions presented, we have chosen to use aggressive preconditioner settings that are particularly advantageous on massively parallel machines. These settings reduce the overhead of global sums and scalar products --- global sums being limited by network latency and local scalar products by memory bandwidth. 
In the preconditioner, we therefore performed 24 SAP cycles with 6 MR iterations for the subdomains. Figure~\ref{13663_13666:fig} shows that for $\kappa = 0.13663$, using 
a small dimension for the deflating subspace, i.e., $k=3$, and keeping the restart length fixed to $m=18$, 
the iteration number decreases from 183 to 159, i.e., by about 15\%. For the larger value of $\kappa = 0.13666$,
deflating 3 vectors halves the iteration numbers from approximately 500 down to 250. Increasing the cycle length does not substantially affect the deflated method. Figure~\ref{13668_136670:fig} shows results for larger values of $\kappa$. For $\kappa = 0.13668$ non-deflated restarted GMRES almost stagnates even if we use 
a relatively large cycle length, $m= 36$. On the other hand, deflating 3 or 6 vectors cures this stagnation 
completely, and we can even reduce the cycle length down to 24 or even 16. A similar observation holds for 
$\kappa = 0.13670$.    

\begin{figure}
\centerline{\includegraphics[width = 0.62\textwidth]{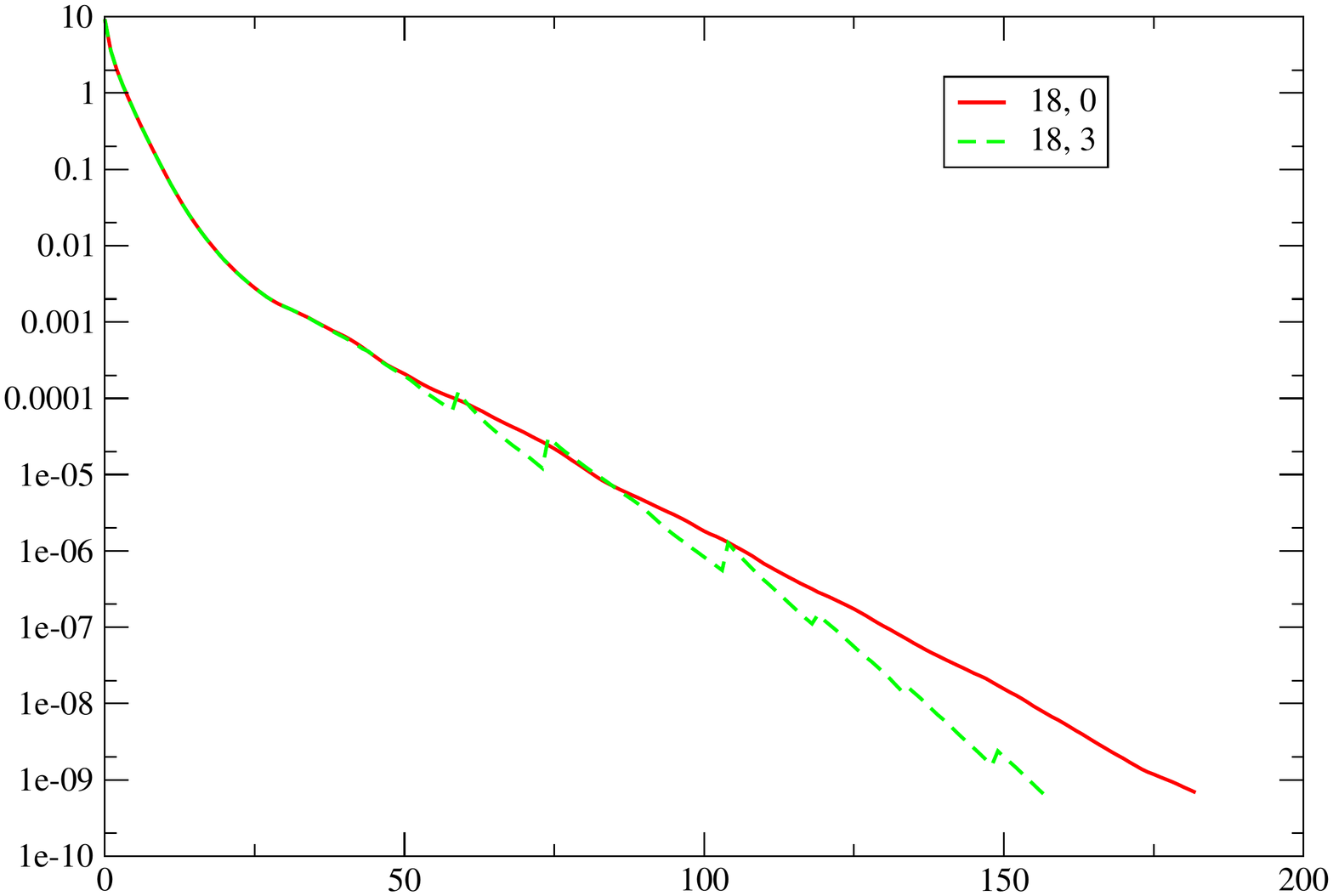} \hfill
\includegraphics[width = 0.62\textwidth]{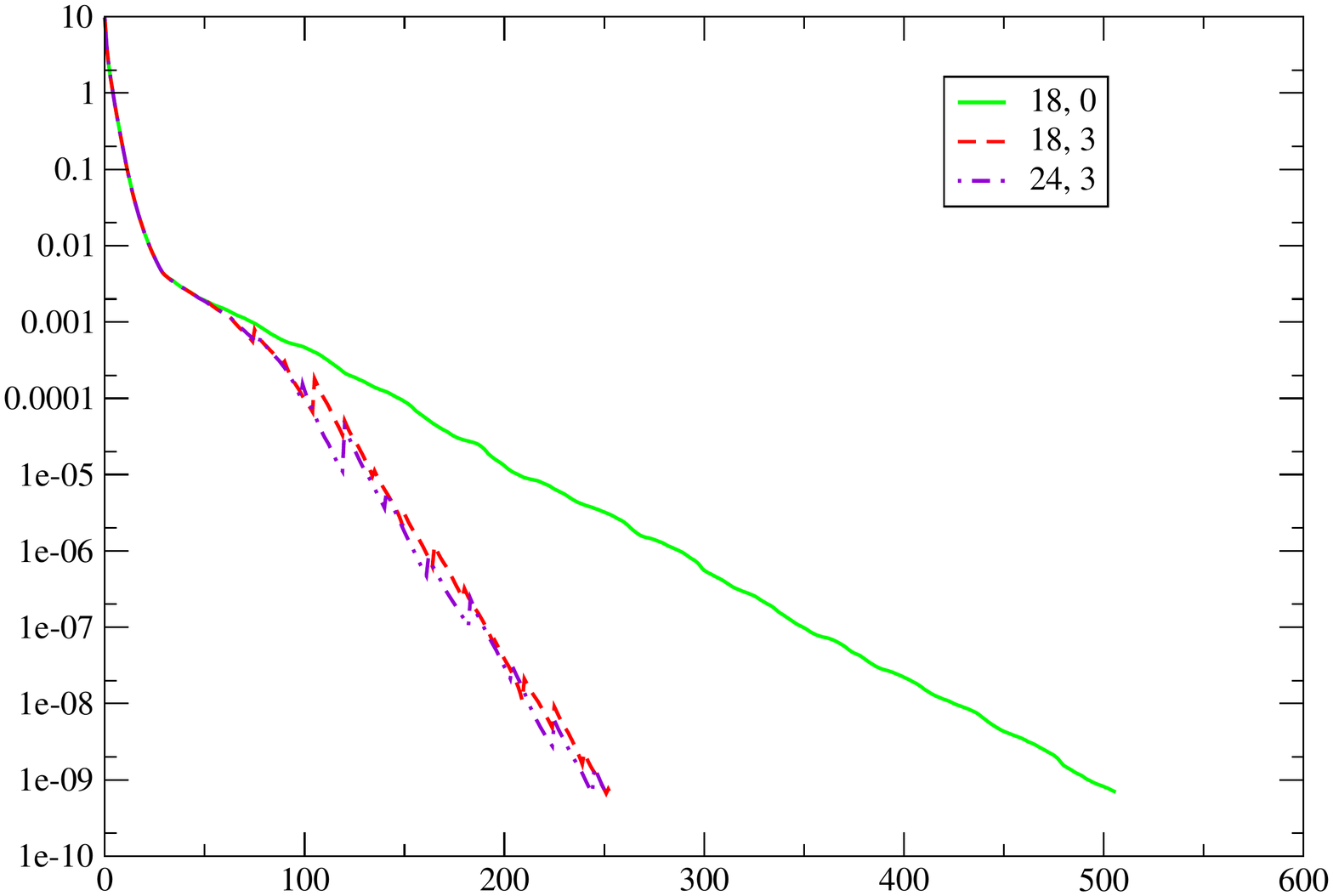}
}
\caption{Results with and without deflation for different solver parameters. $48^3 \times 64$ configuration, $\kappa = 0.13663$ (left),
$\kappa = 0.13666$ (right).\label{13663_13666:fig}} 
\end{figure}

\begin{figure}
\centerline{\includegraphics[width = 0.62\textwidth]{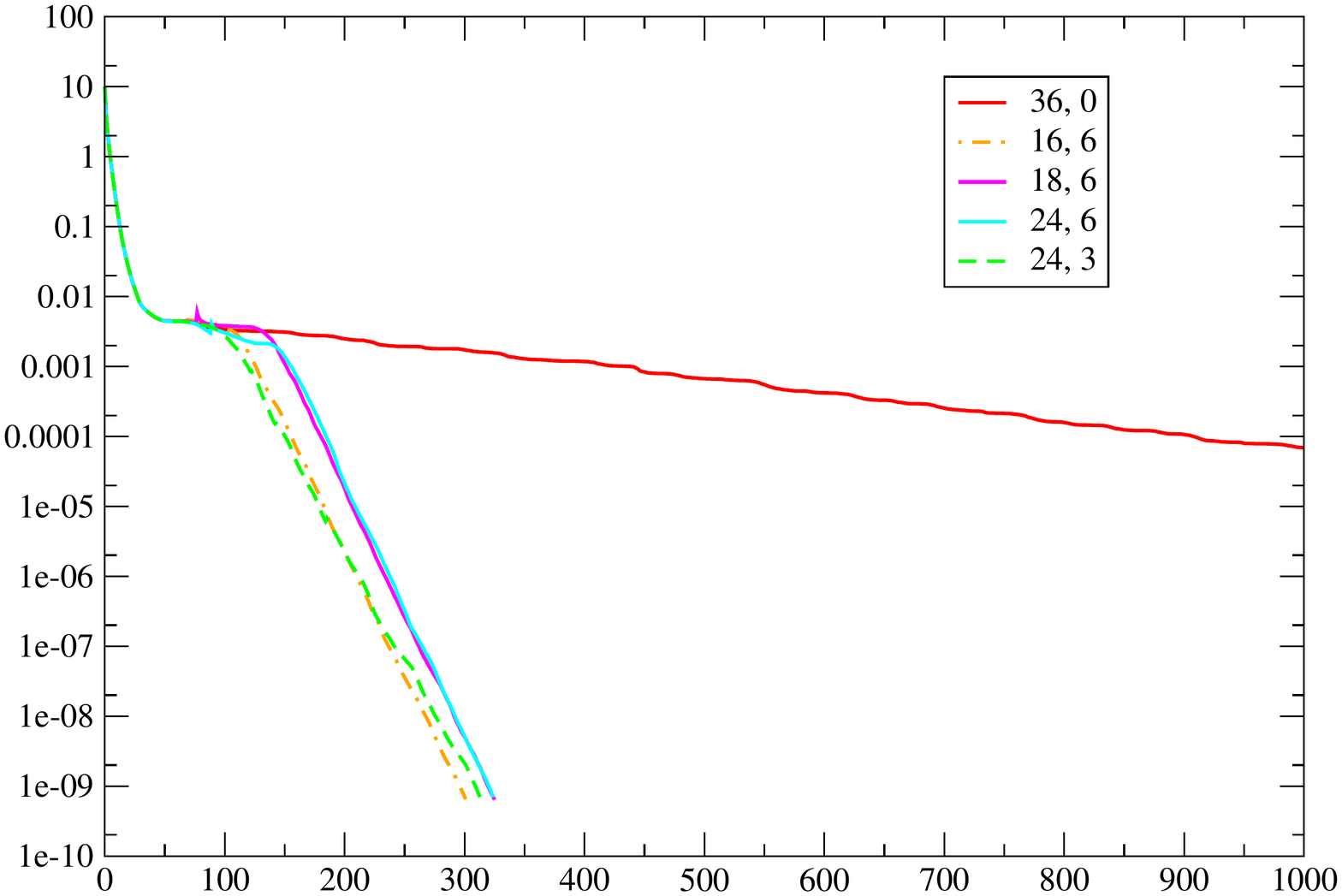} \hfill
\includegraphics[width = 0.62\textwidth]{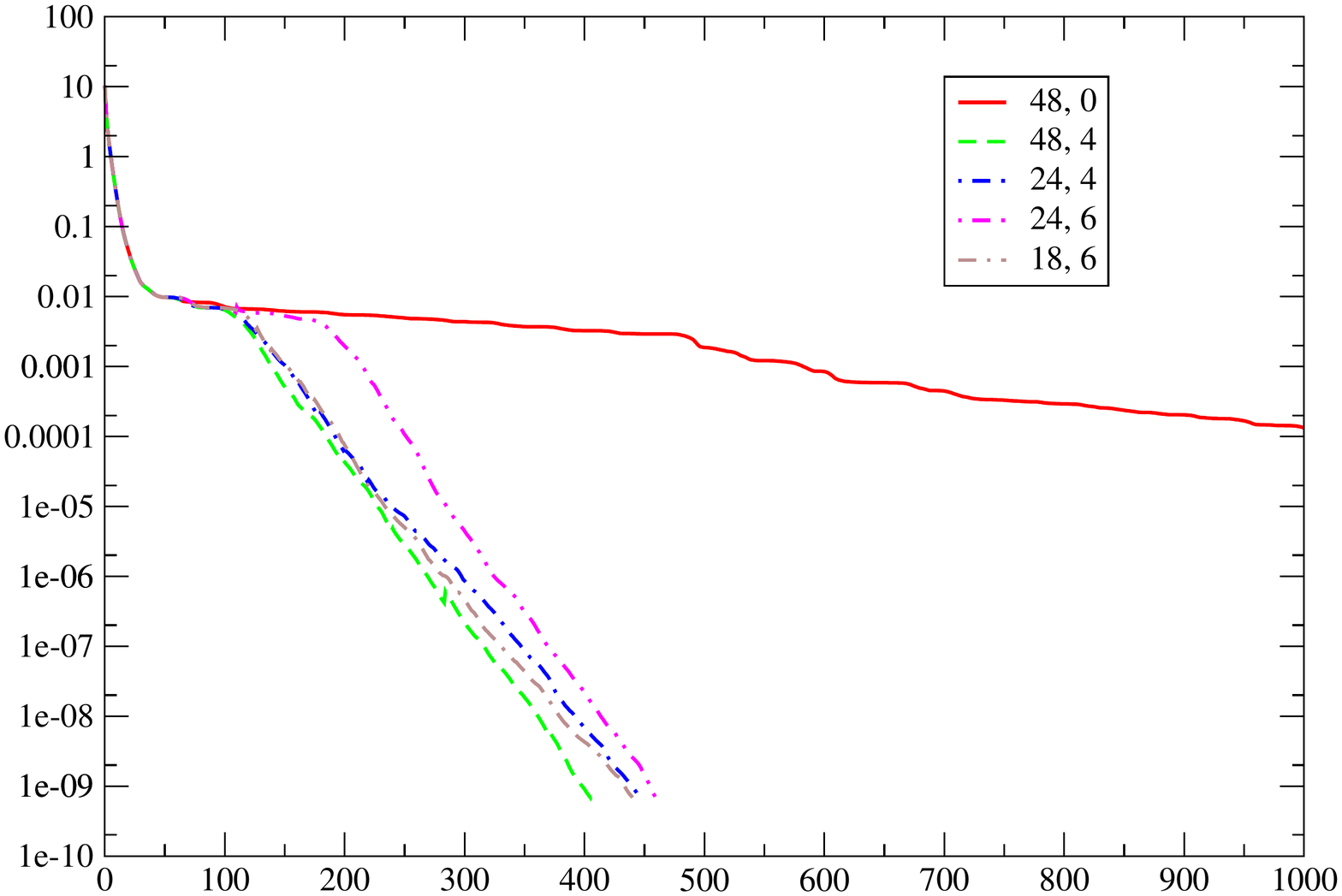}
}
\caption{Results with and without deflation for different solver parameters. $48^3 \times 64$ configuration, $\kappa = 0.13668$ (left),
$\kappa = 0.13670$ (right).\label{13668_136670:fig}} 
\end{figure}

In Table~\ref{table} we report the timings on QPACE for all our computations on $48^3\times 64$ configurations.
These numbers show that wall clock time follows the iteration counts very closely. For the range
tested, the cycle length $m$ by itself has only a marginal influence. This is
due to the fact that our code spends most of its time in the 24 SAP iterations of the preconditioner, so that
the work for the Arnoldi process, which increases with $m$, is relatively cheap for any of our choices for $m$.

  \begin{table}[h]
\begin{center}    
    \begin{tabular}{ | l | c | c | c | c |}
      
    \hline
    kappa & m & k & iterations & time (s) \\ \hline
    0.13663 & 18 & 0 & 183 & 39.0 \\ \hline
    0.13663 & 18 & 3 & 159 & 35.9 \\ \hline
    
    0.13666 & 18 & 0 & 507 & 107.1 \\ \hline
    0.13666 & 18 & 3 & 255 & 56.7 \\ \hline
    0.13666 & 24 & 3 & 252 & 57.4  \\ \hline
    
    0.13668 & 36 & 0 & 3464 & 741.9 \\ \hline
    0.13668 & 16 & 6 & 302 & 70.4 \\ \hline
    0.13668 & 18 & 6 & 326 & 74.9 \\ \hline
    0.13668 & 24 & 6 & 325 & 73.7 \\ \hline
    0.13668 & 24 & 3 & 315 & 69.9 \\ \hline
    
    0.13670 & 48 & 0 & 2869 & 623.6 \\ \hline
    0.13670 & 48 & 4 & 409 & 92.3 \\ \hline
    0.13670 & 24 & 4 & 448 &  100.0\\ \hline
    0.13670 & 24 & 6 & 461 & 106.1 \\ \hline
    0.13670 & 18 & 6 & 441 & 102.0 \\
    \hline

    \end{tabular}
\end{center}
  \caption{Timings and iterations for the $48^3\times 64 $ configuration for the tested $\kappa$ and solver settings combinations \label{table}}  
  \end{table}

\section{Conclusions}
\label{sec:concl}
While SAP can be implemented efficiently on current supercomputers with a deep memory hierarchy and many compute cores like QPACE, it is not necessarily a convergent iteration, particularly if the hopping parameter $\kappa$ 
is close to the critical value. This can very severely
impede the convergence of SAP preconditioned restarted flexible GMRES or restarted GCR. 
We have introduced {\em deflated} restarted flexible GMRES as a (Clover) Wilson fermion solver to cure 
this problem. Numerical results obtained for thermalized configurations of size up to $48^3 \times 64$
show that using a very small dimension for the deflating subspace (never more than 6) always results in
satisfactory convergence speeds. An additional benefit is that deflation allows to use relatively small
cycle lengths (in our experiments, $m=16$ to $24$ was always sufficient), which reduces the arithmetic cost 
of the Arnoldi process and, more importantly, keeps memory requirements low. The benefits of deflation are most prominent
when the hopping parameter $\kappa$ is close to the critical value $\kappa_c$ or even larger than $\kappa_c$,
so that deflation is particularly helpful on exceptional configurations arising during an HMC simulation.

The parameter space on which the performance of SAP preconditioned deflated restarted flexible GMRES 
can be optimized, is very large: the size and shape of the subdomains for SAP, the method and number of iterations
to use for the (approximate) subdomain solves, the number of SAP iterations per preconditioning step, 
the cycle length $m$ and the deflating subspace dimension $k$. We could therefore not perform a complete investigation of this 
parameter subspace.
Rather, our approach was to fix parameters related to the SAP preconditioner by efficiency considerations
for our given hardware, and then find good values for $k$ and $m$. Although this quite surely means that we
missed the overall best method, all our results consistently indicate that it is always worth 
to include a (small) deflating subspace: Convergence is faster for the same cycle length, and for hard problems,
convergence is enabled at all. We thus think that deflated restarted flexible GMRES should establish
itself as the standard successor to flexible restarted GMRES or GCR and even more to preconditioned BiCGstab
for which we observed convergence problems relatively often during simulations for large configurations.      

We note that an alternative approach to the method presented here are multilevel methods like adaptive algebraic
multigrid, see \cite{MGClark2007,MGClark2010_1} or L\"uschers deflated domain decomposition method \cite{Luescher2007} and its true multilevel variants \cite{FrKaKrLeRo11}.
As compared to these approaches, the deflated GMRES method does not require a set-up phase which can be 
relativley costly for the other methods if only one or a few right hand sides have to be solved. It has also 
the advantage of being quite simple to implement and to fit modern supercomputer architectures well. 

\section*{Acknowledgements} 
We would like to thank Dirk Pleiter for his continuing support of this work
and Tilo Wettig and Jaques Bloch for illuminating discussions. We are grateful to Xavier Pinel for 
making his FORTRAN code on FGMRES-DR available to us. 
\section{Bibliography} 
 \bibliographystyle{elsarticle-num}
 \bibliography{Biblio_with_et_al}

\end{document}